\documentclass[final]{siamart1116}
\usepackage{algpseudocode}
\usepackage[caption = false]{subfig}
\usepackage{multirow}
\usepackage{booktabs}
\usepackage{mathtools}

\usepackage[letterpaper,centering]{geometry}
\usepackage{amssymb}

\newcommand{\pic}[2]{\includegraphics[width=#2\textwidth]{#1.pdf}}

\newcommand{\B}[0]{$B_{1,1}^s\;$}
\newcommand{\R}[0]{\mathbb{R}}
\newcommand{\Z}[0]{\mathbb{Z}}
\newcommand{\T}[0]{\top}
\newcommand{\N}{\mathcal{N}}
\newcommand{\mepow}[2]{#1\hspace{-0.12em}\cdot\hspace{-0.12em}10^{#2}}
\newcommand{\epow}[2]{#1\cdot10^{#2}}
\newcommand{\half}{\frac{1}{2}}
\newcommand{\spaces}{\hspace{2 em}}
\newcommand{\defeq}{\coloneqq}
\newcommand{\gd}{g_\text{1D}}
\newcommand{\J}[1]{J_{#1}}
\newcommand{\Var}{\mathrm{Var}}

\DeclareMathOperator*{\argmin}{arg\,min}

\renewcommand{\thetheorem}{\thesection.\arabic{theorem}}
\renewcommand{\theequation}{\thesection.\arabic{equation}}
\renewcommand{\thefigure}{\thesection.\arabic{figure}}
\renewcommand{\thetable}{\thesection.\arabic{table}}

\newcommand{\resetcounters}{
  \setcounter{theorem}{0}
  \setcounter{equation}{0}
  \setcounter{figure}{0}
  \setcounter{table}{0}
  \setcounter{algorithm}{0}
}

\newsiamremark{rem}{Remark}
\newsiamthm{thm}{Theorem}
\newsiamthm{lem}{Lemma}
\newsiamthm{assume}{Assumption}

\newcounter{subassume}[theorem]
\renewcommand{\thesubassume}{(\textit{\roman{subassume}})}
\makeatletter
\renewcommand{\p@subassume}{\thetheorem}
\makeatother
\newcommand{\assumeitem}{
  \refstepcounter{subassume}%
  \makebox[2em][l]{\thesubassume}\ignorespaces}

\title{Bayesian inverse problems with $l_1$ priors: a
  randomize-then-optimize approach}

\makeatletter
\newcommand{\algmargin}{\the\ALG@thistlm}
\makeatother
\newlength{\whilewidth}
\algdef{SE}[parWHILE]{parWhile}{EndparWhile}[1]
  {\parbox[t]{\dimexpr\linewidth-\algmargin}{%
     \hangindent\whilewidth\strut\algorithmicwhile\ #1\ \algorithmicdo\strut}}{\algorithmicend\ \algorithmicwhile}%
\algnewcommand{\parState}[1]{\State%
  \parbox[t]{\dimexpr\linewidth-\algmargin}{\strut #1\strut}}

\author{Zheng Wang%
  \thanks{Department of Aeronautics and Astronautics, MIT, Cambridge, MA 02139 (\email{zheng\_w@mit.edu}, \email{ymarz@mit.edu}.)}%
  \and
  Johnathan M.~Bardsley%
  \thanks{Department of Mathematical Sciences, Montana, University of Montana, Missoula, MT 59812 (\email{bardsleyj@mso.umt.edu}.)}
  \and
  Antti Solonen%
  \thanks{Department of Mathematics and Physics, Lappeenranta University of Technology, Lappeenranta, Finland (\email{antti.solonen@gmail.com}.)}
  \and
  Tiangang Cui%
  \thanks{School of Mathematical Sciences, Monash University, Victoria 3800, Australia (\email{tiangang.cui@monash.edu}.)} 
  \and
  Youssef~M. Marzouk%
  \footnotemark[1]
}

\begin{document}

\maketitle

\begin{abstract}
Prior distributions for Bayesian inference that rely on the $l_1$-norm of the parameters are of considerable interest, in part because they promote parameter fields with less regularity than Gaussian priors (e.g., discontinuities and blockiness). These $l_1$-type priors include the total variation (TV) prior and the Besov space \B prior, and in general yield non-Gaussian posterior distributions. Sampling from these posteriors is challenging, particularly in the inverse problem setting where the parameter space is high-dimensional and the forward problem may be nonlinear.
This paper extends the randomize-then-optimize (RTO) method, an optimization-based sampling algorithm developed for Bayesian inverse problems with Gaussian priors, to inverse problems with $l_1$-type priors. We use a variable transformation to convert an $l_1$-type prior to a standard Gaussian prior, such that the posterior distribution of the transformed parameters is amenable to Metropolized sampling via RTO. We demonstrate this approach on several deconvolution problems and an elliptic PDE inverse problem, using TV or Besov space \B priors. Our results show that the transformed RTO algorithm characterizes the correct posterior distribution and can be more efficient than other sampling algorithms. The variable transformation can also be extended to other non-Gaussian priors.
\end{abstract}



\section{Introduction}
Inverse problems are encountered in many fields of science and engineering---whenever unknown parameters in mathematical models of physical phenomena must be estimated from noisy, incomplete, and indirect measurements.
While inverse problems can be solved using a variety of approaches \cite{Tarantola}, the Bayesian statistical approach \cite{KaipioSomersalo,StuartActa} is particularly attractive as it offers a coherent framework for quantifying parameter uncertainty, while naturally accommodating different types of data and rich models of prior information.

We begin our discussion of the Bayesian approach to inverse problems by considering a parametric statistical model of the form
\begin{equation}
\label{eq:model}
y=f(\theta)+\epsilon,
\end{equation}
where $y\in\R^m$ is a vector of measurements, $f:\R^n\rightarrow\R^m$ is the forward model (also known as the ``parameter-to-observable map'') relating the unknown parameters $\theta\in\R^n$ to the measurements $y$, and $\epsilon\in\R^m$ is the measurement error. We assume that the error $\epsilon$ is a Gaussian random vector with mean zero and covariance matrix $\Gamma_{\rm obs}\in \R^{m\times m}$, i.e.,  $\epsilon\sim\N(0,\Gamma_{\rm obs})$. We will consider both linear and nonlinear forward models $f$. 

Next, define a prior probability density,
\[p(\theta)\propto\exp(-\lambda J(\theta)),\]
that encapsulates all \textit{a priori} information on the parameters $\theta$. Here, $\lambda \in \R$ is a hyperparameter and $J:\R^n\rightarrow\R$ is a prescribed function. Through Bayes' rule, the prior density and the likelihood function defined by \eqref{eq:model} together yield the posterior probability density of the parameters $\theta$:
\begin{equation}
\label{eq:posterior0} 
p(\theta|y) \propto p(y|\theta)p(\theta) \propto \exp\left(-\half\big\| f(\theta)-y\big\|^2_{\Gamma_{\rm obs}^{-1}}-\lambda J(\theta)\right).
\end{equation}

Solving the inverse problem in the Bayesian setting amounts to characterizing the posterior distribution \eqref{eq:posterior0}, e.g., computing posterior moments or other posterior expectations. A flexible way to do so is via sampling, which has been a topic of research in Bayesian inverse problems for decades (see, e.g., \cite{CalSom,KaipioSomersalo,NicFox,StuartActa}). A widely used class of algorithms for sampling from the posterior is Markov chain Monte Carlo (MCMC); see, e.g., \cite{bayesiandata,mcstatmodel,mcscientific,mcmcinpractice,GamLop} for a general introduction. Most MCMC algorithms build on the Gibbs \cite{gibbs} or general Metropolis-Hastings \cite{metropolis,hastings,green} constructions. For example, \cite{KaiKolSomVau,KaipioSomersalo} implement Gibbs samplers for use on large-scale  nonlinear inverse problems, while \cite{adaptm,dram} introduce adaptive Metropolis algorithms that work well on parameter inference problems of small to medium dimension. The need for adaptive algorithms underscores the idea that efficient MCMC sampling requires proposal distributions that capture the local or global structure of the target (posterior) distribution. Accordingly, the Metropolis-adjusted Langevin algorithm (MALA) \cite{mala} uses gradients of the target density to guide samples towards regions of higher probability, while \cite{MarWilBurGha} approximates local Hessians of the log-target density to construct Gaussian proposals for large-scale problems. Riemannian manifold MCMC \cite{riemann} may use even higher-order derivative information, along with Hamiltonian Monte Carlo (HMC) \cite{hmc,nuts} proposals. Another issue, particularly relevant to Bayesian inverse problems where $\theta$ represents the discretization of a distributed parameter, is that most MCMC algorithms have mixing rates that deteriorate as the discretization is refined \cite{convrwm,convrates,convdifflim}. Recent work \cite{pCN} has introduced Metropolis algorithms with discretization-invariant mixing properties. Dimension-independent likelihood-informed (DILI) samplers then combine discretization invariance with proposals informed by Hessians and other descriptors of the posterior geometry \cite{dili}. With the exception of HMC, however, even these relatively sophisticated samplers produce Gaussian proposals at each step. From a computing perspective, we also note that most MCMC algorithms are sequential in nature and may not scale well to massively parallel settings (e.g., via multiple chains) \cite{badparallelGeyer}.

This paper builds on recent work that explores the potential for \textit{optimization methods} to improve sampling. Broadly, these methods facilitate simulation from non-Gaussian proposal distributions that capture important aspects of posterior structure. 
Notable examples include randomized maximum likelihood \cite{rml}, implicit sampling \cite{generalimplicit,implicit}, and randomize-then-optimize (RTO) \cite{rto}. Our focus in this work is on the RTO approach. RTO uses repeated solutions of a randomly perturbed optimization problem to produce samples from a non-Gaussian distribution, which is used as a Metropolis independence proposal. Although it is more expensive to implement per sample than many simpler Gaussian proposals, it often yields better MCMC mixing. In addition, because the proposals can be generated independently and in parallel, RTO can easily take advantage of large-scale parallel computing environments. However, RTO is only defined for certain classes of problems; in the case of Bayesian inverse problems, it is defined for problems with Gaussian priors and Gaussian measurement error.

The main contribution of this paper is to extend RTO to \textit{non-Gaussian priors}, and to understand the efficiency of the resulting posterior sampling algorithm. We will focus on the case of $l_1$-type priors, but the approach can be used on other priors as well. In using $l_1$-type priors, we assume that there is a deterministic invertible matrix $D \in \R^{n\times n}$, such that the elements of the vector $D\theta$ are \textit{a priori} independent and endowed with identical Laplace distributions. Thus, the prior is of the form
\begin{equation}
\label{eq:l1prior}
p(\theta) \propto  \exp{\left(-\lambda \sum_{i=1}^n |(D\theta)_i | \right)},
\end{equation}
where $\lambda \in \R$ is a hyperparameter. This choice yields a posterior of the form
\begin{equation}
p(\theta|y) \propto \exp\left(-\half\Vert f(\theta)-y\Vert^2_{\Gamma_{\rm obs}^{-1}}-\lambda \sum_{i=1}^n |(D\theta)_i |\right).\label{eq:posterior}
\end{equation}
For what is perhaps the most common $l_1$-type prior used in Bayesian inverse problems, $D$ is the discrete one-dimensional derivative (or difference) matrix. This choice yields the total variation (TV) prior, which is related to the well-known regularization functional that penalizes the variation of a signal in order to promote a blocky, discontinuous solution \cite{tv,Vogel}.  The TV prior can be derived from the assumption that the increments (i.e., the differences between neighboring parameter node values) are i.i.d.\ Laplace random variables \cite{BarLaplace}, and it has the form \eqref{eq:l1prior} only when $\theta$ is the discretization of a one-dimensional signal. Another common class of $l_1$-type priors are the Besov space \B priors \cite{discretization-invariant}, where $D$ is now a matrix representing a discrete wavelet transform \cite{daubechies}; for the use of Besov priors on large-scale imaging test cases, see \cite{DauDefDeM,MueSil}. These priors have the advantage that even in two or more dimensions, they retain the form \eqref{eq:l1prior} and hence the techniques of this paper can be used. Besov priors (with suitable parameters) have been shown to be discretization invariant \cite{discretization-invariant,besov}, in that they yield posterior means that converge under mesh refinement. 

We extend RTO to the problem of sampling from \eqref{eq:posterior} by introducing a multivariate ``prior transformation.'' This transformation deterministically couples a random variable with an $l_1$-type prior to one with a Gaussian prior, and thus enables the use of RTO. A similar transformation for a scalar parameter $\theta$ has been suggested in \cite{Oliver_arXiv:1507.08563}. The present multivariate transformation is more general, however. To the best of our knowledge, it has not been previously proposed, nor has its impact on sampling been investigated.
After modifying the RTO algorithm to incorporate the transformation, we conduct a simple comparison of the resulting method with other algorithms, and then focus on numerically exploring the factors that influence its efficiency.

More broadly, variable transformations have been used to improve sampling in \cite{vartrans,mattMCMC}. For instance, \cite{mattMCMC} learns a parameterized multivariate transformation, designed to approximately Gaussianize an arbitrary target distribution, adaptively during MCMC. \cite{vartrans} introduces fixed isotropic (i.e., $\Vert \theta \Vert$--dependent) transformations to obtain target distributions with super-exponentially light tails, so that random-walk Metropolis sampling is geometrically ergodic. In a similar fashion, we use our prior transformation to obtain a posterior distribution to which we can apply RTO. We also describe extensions of our approach to more general priors: first, when any \textit{exact} (e.g., closed-form) coupling between the prior and a standard Gaussian is available, and second, when the prior transformation is only \textit{approximate}. In the latter case, we modify the Metropolis step of our RTO sampler to correct for error in the prior transformation.

The remainder of the paper is organized as follows. We begin in Section~\ref{sec:rto} with a description of the RTO algorithm \cite{rto}. Then, in Section~\ref{sec:trans}, we describe prior transformations that turn \eqref{eq:posterior} into a target density amenable to RTO sampling. Finally, in Section~\ref{sec:ex}, we present several numerical examples and comparisons of our method with other MCMC algorithms.

\resetcounters

\section{Randomize-then-optimize} \label{sec:rto}
In the context of Bayesian inverse problems, the randomize-then-optimize (RTO) \cite{rto} algorithm can be used to sample from the posterior distribution if the prior distributions on the parameters $\theta$ and the measurement error $\epsilon$ are both Gaussian. It generates proposal samples through optimization, and then ``corrects'' these samples using either importance sampling or Metropolis-Hastings. Here, we briefly review the original RTO algorithm; for simplicity, we use notation slightly different from that of \cite{rto}.

\subsection{Form of the target distribution} 
RTO requires that the target distribution be of a specific form; in particular, it requires that the target density (which for the purposes of this paper is the posterior density of $\theta$) be written as
\begin{equation}  
p(\theta\vert y ) \propto \exp\left(-\half \big\Vert F(\theta) \big\Vert ^2\right), \label{eq:lsform}
\end{equation}
where $F(\theta)$ is a vector-valued function of the parameters $\theta$.

Given a Gaussian prior and Gaussian measurement errors, we can, without loss of generality, use linear transformations to ``whiten'' the prior and the error model so that the inverse problem has the form
\begin{equation} 
y = f(\theta) + \epsilon, \spaces \epsilon \sim N(0,I_m), \spaces \theta \sim N(\theta_0 ,I_n), \label{eq:white_form}
\end{equation}
where $\theta_0 \in \R^n$ is the prior mean; and $I_n$ and $I_m$ are identity matrices of size $n$ and $m$, respectively. The resulting posterior density is given by
\[
p(\theta\vert y ) \propto \exp{\Bigg(-\half\bigg\Vert \begin{bmatrix} \theta\\ f(\theta) \end{bmatrix} - \begin{bmatrix} \theta_0 \\ y \end{bmatrix}\bigg\Vert ^2\Bigg)}.
 \]
This density is in the form \eqref{eq:lsform}, where $F(\theta)=\left[\begin{array}{c}\theta - \theta_0 \\f(\theta) - y \end{array}\right]$ and $F :\R^n\rightarrow\R^{n+m}$.

\subsection{The RTO--Metropolis-Hastings algorithm}

\newcommand{\htheta}{\theta_\text{prop}}
\newcommand{\hthetai}{\theta_\text{prop}^{(i)}}

We now outline how to use RTO to sample from a posterior of the form \eqref{eq:lsform}. First, a linearization point $\bar{\theta}$ is found and fixed throughout the algorithm. In \cite{rto}, $\bar{\theta}$ is set to be the posterior mode, though this is not the only possible or useful choice. To obtain the posterior mode, we solve
\begin{equation} \label{eq:opmode} \overline{\theta} = \argmin_{\theta} \half \big\Vert  F(\theta) \big\Vert ^2.\end{equation}
Second, the Jacobian of $F$, which we denote as $\J{F}$, is evaluated at $\overline{\theta}$, and an orthonormal basis $\overline{Q} \in \R^{(m+n)\times n}$ for the column space of $\J{F}(\overline{\theta})$, which we denote as $\mathrm{col}(\J{F}(\overline{\theta}))$, is computed through a thin-QR factorization of $\J{F}(\overline{\theta})$. Third, independent samples $\xi^{(i)}$ are drawn from an $n$-dimensional standard Gaussian, and proposal points $\hthetai$ are found by solving the optimization problem
\begin{equation} \label{eq:optprob}
\hthetai = \argmin_{\theta} \half \left\Vert  \overline{Q}^\T F(\theta) - \xi^{(i)} \right\Vert ^2  \end{equation}
for each sample  $\xi^{(i)}$. Under conditions described in \cite{rto} and listed in Assumption~\ref{assump}, the points $\hthetai$ are distributed according to the proposal density
\begin{equation}
q(\htheta) = (2 \pi)^{-\frac{n}{2}} \left\vert \overline{Q}^\T \J{F}(\htheta)\right\vert  \exp\left(-\frac12 \Big\Vert  \overline{Q}^\T F(\htheta) \Big\Vert ^2\right),\label{eq:RTOproposal}
\end{equation}
where $\left\vert \cdot \right\vert $ denotes the absolute value of the matrix determinant. We focus on using this distribution as an independence proposal in Metropolis-Hastings, though it can also be used in importance sampling. The Metropolis-Hastings acceptance ratio, for a move from a point $\theta^{(i-1)}$ to the proposed point $\hthetai$, is
\[
\frac{p(\hthetai\vert y )q(\theta^{(i-1)})}{p(\theta^{(i-1)}\vert y )q(\hthetai)} =\frac{w(\hthetai)}{w(\theta^{(i-1)})},
\]
where $w(\theta)$ are
\begin{equation}\label{eq:w(theta)}
w(\theta) \defeq \left\vert \overline{Q}^\T \J{F}(\theta)\right\vert ^{-1} \exp\left(-\frac12 \big\Vert F(\theta)\big\Vert ^2+\frac12 \Big\Vert  \overline{Q}^\T F(\theta) \Big\Vert ^2\right).
\end{equation}
The resulting MCMC method, which we call RTO--Metropolis-Hastings (RTO-MH), is summarized in Algorithm~\ref{alg:rto}.

\begin{algorithm}[htbp]
\caption{RTO-MH} \label{alg:rto}
\begin{algorithmic}[1]
\State Find $\overline{\theta}$ (e.g., the posterior mode) using \eqref{eq:opmode}
\State Determine $\J{F}(\overline{\theta})$, the Jacobian of $F$ at $\overline{\theta}$
\State Compute $\overline{Q}$, whose columns are an orthonormal basis for $\mathrm{col}(\J{F}(\overline{\theta}))$
\For {$i = 1, \ldots, n_\textrm{samps}$} in parallel
  \State Sample $\xi^{(i)}$ from a standard $n$-dimensional Gaussian
  \State Solve for a proposal sample $\hthetai$ using \eqref{eq:optprob}
  \State Compute $w(\hthetai)$ from \eqref{eq:w(theta)}
\EndFor

\State Set $\theta^{(0)} = \overline{\theta}$
\For {$i = 1, \ldots, n_\textrm{samps}$} in series
\State Sample $v$ from a uniform distribution on [0,1]
\If {$v < \left. w(\hthetai) \right/ w(\theta^{(i-1)})$}
  \State $\theta^{(i)}$ = $\hthetai$
\Else
  \State $\theta^{(i)}$ = $\theta^{(i-1)}$
\EndIf
\EndFor
\end{algorithmic}
\end{algorithm}

\begin{rem} Other choices for the matrix $\overline{Q}$ used in \eqref{eq:optprob} and \eqref{eq:RTOproposal} are possible, provided that
Assumption~\ref{assump}, which leads to the sampling density $q(\theta)$ in \eqref{eq:RTOproposal}, is satisfied. 
Also, in the computation of the Metropolis acceptance ratio, one can use a factorization of $\J{F}(\theta)$ or $\overline{Q}^\T \J{F}(\theta)$ and take advantage of properties of the $\log$ function; e.g., if $Q_\theta R_\theta=\overline{Q}^\T \J{F}(\theta)$ is the $QR$ factorization of $\overline{Q}^\T \J{F}(\theta)$, then
\[
\log \left\vert \overline{Q}^\T \J{F}(\theta)\right\vert =\sum_{i=1}^n\log [R_\theta]_{ii}.
\]
\end{rem}

\resetcounters

\section{RTO-MH with a prior transformation} \label{sec:trans}
The previous section showed how Bayesian inverse problems with Gaussian priors and Gaussian measurement errors yield posterior densities that can be written in the form \eqref{eq:lsform}, as required by RTO. Now we propose a technique that uses RTO to sample from a posterior resulting from a Gaussian measurement model and a \textit{non-Gaussian} prior. This is accomplished via a change of variables that transforms the non-Gaussian prior defined on the physical parameter $\theta \in \R^n$ to a standard Gaussian prior defined on a reference parameter $u \in \R^n$. The caveat is that the transformed forward model, now viewed as a function of $u$, is the original forward model composed with the nonlinear mapping function, and hence the transformation adds complexity to $f$. 

\subsection{Transformations for \texorpdfstring{$l_1$}{l1}-type priors}

In the following subsections, we exemplify this approach for $l_1$-type priors. First, we describe the transformation of single parameter endowed with a Laplace prior (Section~\ref{sec:1d}). We then extend that example to construct a transformation of multiple parameters for any $l_1$-type prior (Section~\ref{sec:md}). Finally, we discuss general prior transformations and summarize the algorithm for performing RTO with a prior transformation (Section~\ref{sec:rtoprior}).

\subsubsection{Single parameter with a Laplace prior} \label{sec:1d}
In this subsection, we consider an inverse problem of the form \eqref{eq:model} but with only a single parameter and a single observation, $n=m=1$:
\[ y = f(\theta) + \epsilon, \spaces \epsilon \sim N(0,\sigma_{\textrm{obs}}^2),\]
where $\sigma_{\textrm{obs}} \in \R^+$ is the standard deviation of the error. Instead of a Gaussian prior on $\theta$, we use a Laplace prior
\[ p(\theta) \propto \exp{\left( - \lambda  \vert\theta \vert \right).} \]
Then, the posterior has the form
\begin{equation}\label{eq:1Dpost}
 p(\theta \vert y) \propto \exp{\left(- \half\left( \frac{f(\theta)-y}{\sigma_{\textrm{obs}}} \right)^2 - \lambda  \vert \theta \vert \right)}. \end{equation}
Due to the Laplace prior, $p(\theta \vert y)$ cannot directly be written in the form \eqref{eq:lsform}.

Let us construct an invertible mapping function $\gd : \R \rightarrow \R$ that relates a Gaussian reference random variable $u \in \R$ to the Laplace-distributed physical parameter $\theta \in \R$, such that $\theta = \gd(u)$. A monotone transformation that achieves this goal is
\begin{equation}\label{eq:g1D} 
\gd(u) \equiv \mathcal{L}^{-1}\left(\varphi(u)\right) = -\frac{1}{\lambda} \mathrm{sign}\left(u\right) \log{\left(1 - 2 \left \vert \varphi(u) -\half \right \vert \right)},
\end{equation}
where $\mathcal{L}$ is the cumulative distribution function (cdf) of the Laplace distribution and $\varphi$ is the cdf of the standard Gaussian distribution.
To prove that the reference random variable is in fact standard Gaussian, we calculate its cdf as:
\begin{align*}
\mathbb{P}(u<u_0) &= \mathbb{P}(\gd^{-1}(\theta) < u_0) = \mathbb{P}(\theta < \gd(u_0)) \\
&= \mathcal{L}(\gd(u_0)) = \mathcal{L}\circ\mathcal{L}^{-1}\circ\varphi(u_0) = \varphi(u_0).
\end{align*}
Hence, this mapping function indeed transforms a standard Gaussian reference random variable $u$ to the Laplace-distributed parameter $\theta$, and thus
$p(u) \propto \exp{\left(-\half u^2 \right)}.$

 The mapping function $\gd$ and its derivative are depicted in Figure~\ref{fig:map}. The function is monotone, bijective, and continuously differentiable. Its derivative is
\[ \gd'(u) = \frac {\varphi'(u)}{\lambda \varphi( - \vert u \vert)}, \]
where $\varphi'(u) = \frac{1}{\sqrt{2\pi}} \exp{\left(-\half u^2\right)}$ is the probability density function of the standard Gaussian distribution.

\begin{figure}[htbp]
\centering
\subfloat[Construction of $\gd$ from densities.]{\pic{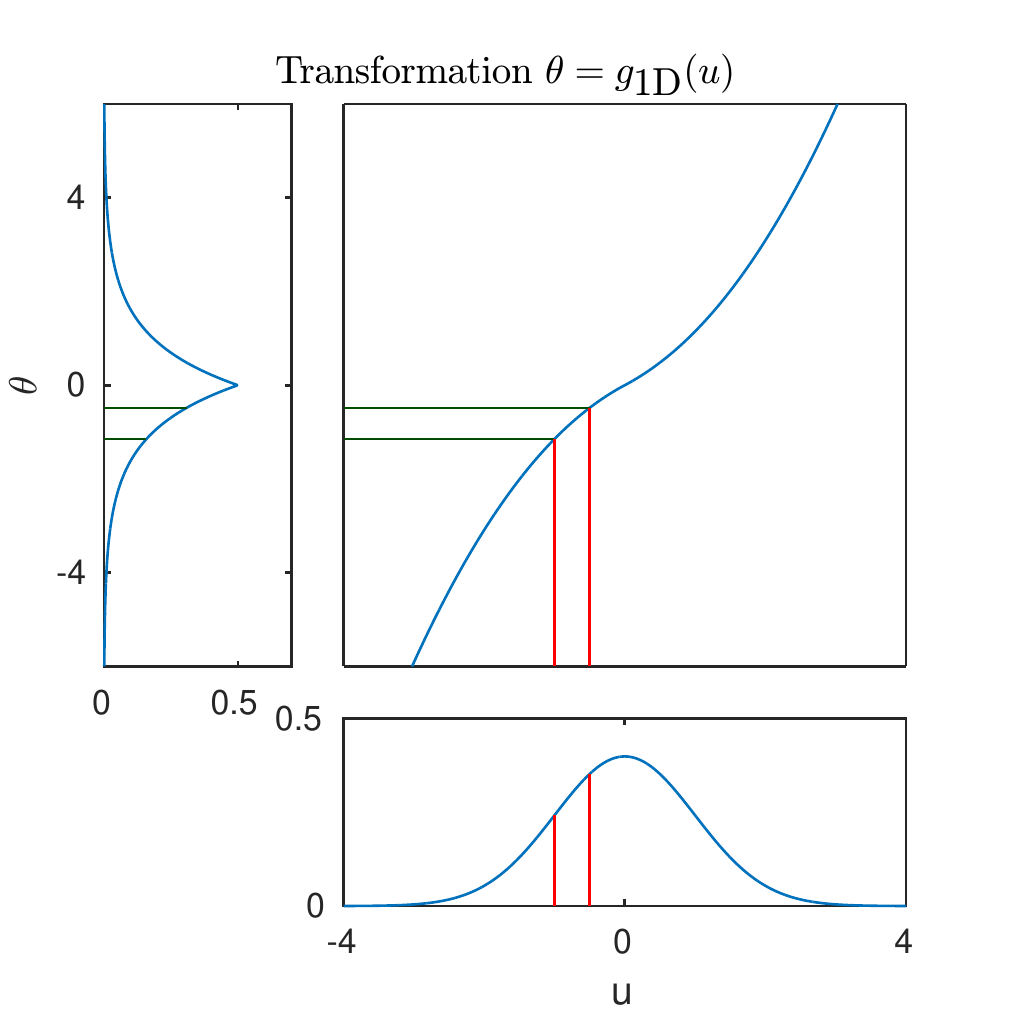}{0.5}} 
\subfloat[$\gd$ and its derivative.]{\pic{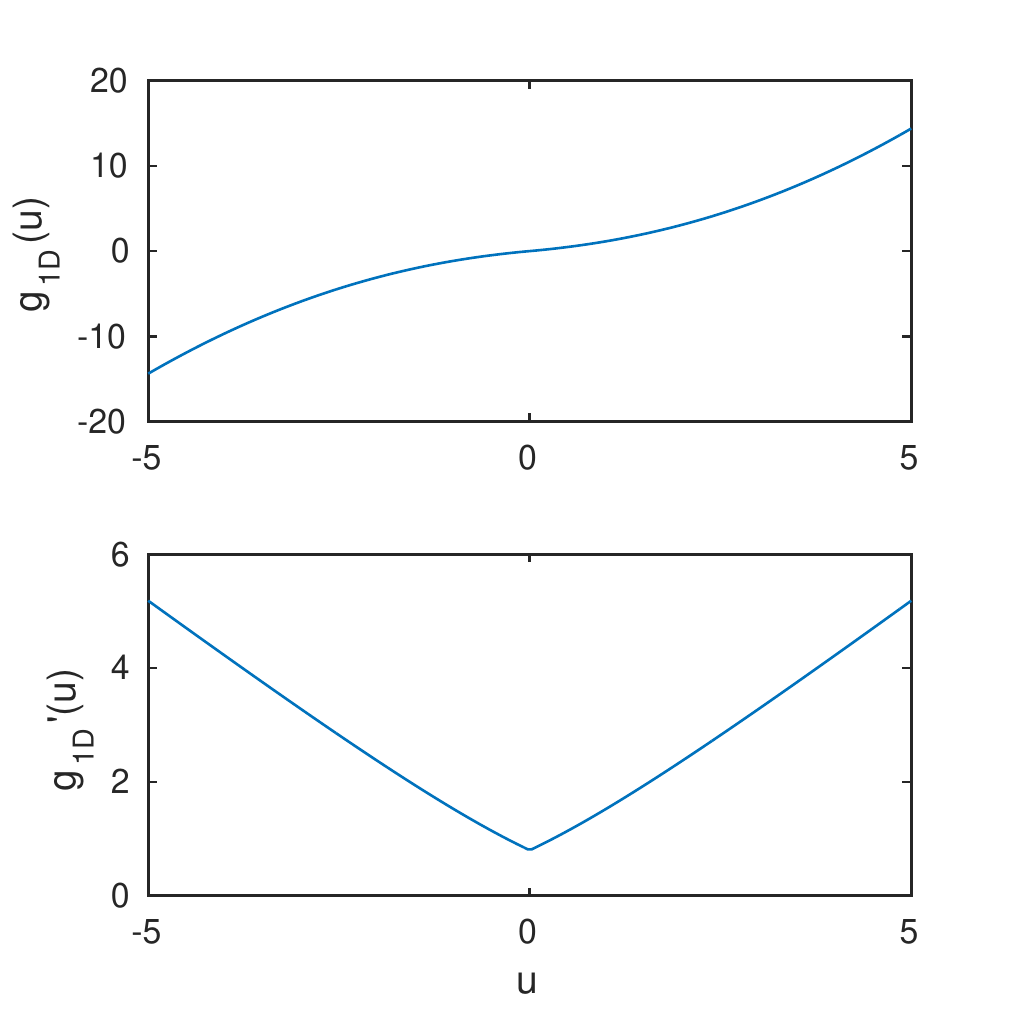}{0.49}}
\caption[Transformation of Gaussian to Laplace]{Left: transformation from the standard Gaussian to a Laplace distribution (with $\lambda = 1$). The probability mass between the two vertical lines is equal to that between the horizontal lines. Right: Mapping function $\gd$ and its derivative. The mapping function is continuously differentiable.}\label{fig:map}
\end{figure}

Now we can solve Bayesian inverse problems on $u$ and transform the posterior samples of $u$ to posterior samples of $\theta$ using the mapping function. The form of the \textit{transformed} posterior density, i.e., the posterior density of $u$, is given in the following lemma and proven in Appendix~\ref{sec:proofpost}.

\begin{lem}\label{lem:post}
  Let \eqref{eq:1Dpost} specify the posterior density of a Bayesian inference problem with parameter $\theta \in \R$. Given the variable transformation $\theta = \gd(u)$ defined in \eqref{eq:g1D}, the posterior density of $u$ has the form:
  \begin{equation}\label{eq:trpost}
    p(u\vert y) \propto \exp{\left( - \frac12 \left(\frac{f\circ \gd(u)-y}{\sigma_\text{obs}}\right)^2 - \frac12 u^2 \right)}.
  \end{equation}
\end{lem}

After the transformation, the prior over the new variables simplifies to a standard Gaussian, and the forward model becomes more complex. In particular, the transformed forward model is the original forward model composed with the nonlinear mapping. The new posterior appears with a Gaussian prior and observational noise, and can be cast in the form \eqref{eq:lsform}. The resulting structure allows us to use RTO.

\subsubsection{Multiple parameters with an \texorpdfstring{$l_1$}{l1} prior} \label{sec:md}
Now we build on the previous section in order to construct a prior transformation for a multivariate $l_1$-type prior. Starting from an inverse problem of the form \eqref{eq:model}, we allow for multiple unknown parameters, $n\ge1$, and multiple observations, $m\ge1$. We impose the following $l_1$-type prior on $\theta$:
\[ p(\theta) \propto \exp{\Big( -\lambda \Vert  D \theta \Vert _1 \Big)} = \exp{\left(-\lambda \sum_{i=1}^n |(D\theta)_i | \right),} \]
where $D \in \R^{n \times n}$ is an invertible matrix and $(D\theta)_i$ denotes the $i$th element of vector $D\theta$. The posterior on $\theta$ is then
\begin{equation}\label{eq:l1post}
 p(\theta\vert y) \propto \exp{\left[- \half(f(\theta)-y)^\T  \Gamma_{\textrm{obs}}^{-1} (f(\theta)-y) \right]} \exp{\Big( - \lambda \Vert  D \theta \Vert _1 \Big).} \end{equation}

\begin{figure}[tbp]
\centering
\pic{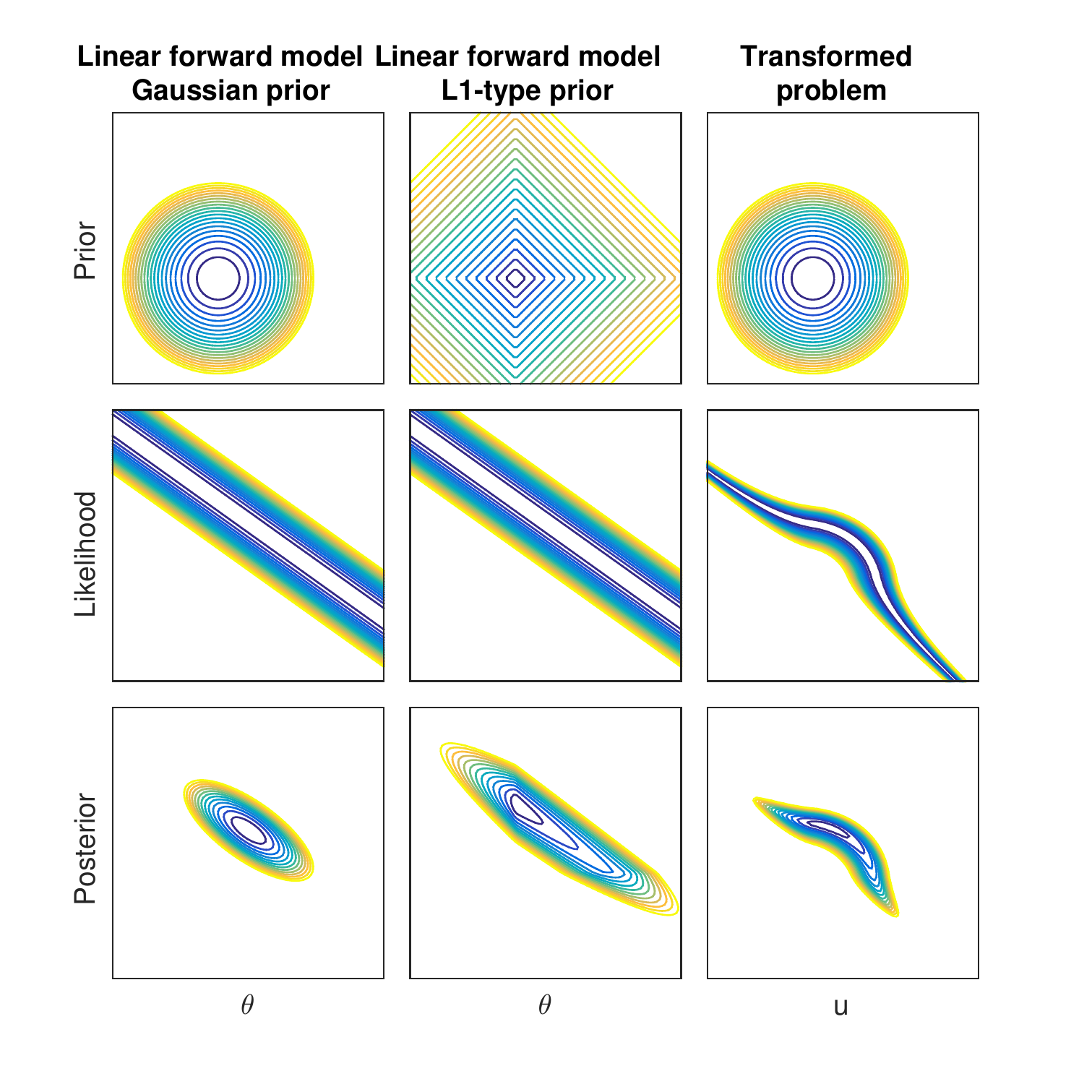}{0.75}
\vspace{-2 em}
\caption[Posteriors]{Densities in illustrative inverse problems with two parameters. The plots depict the log-prior density, log-likelihood function, and log-posterior density. The three cases shown are a Gaussian prior with a linear forward model (left), $l_1$-type prior with the same forward model (middle), and transformed $l_1$-type prior with transformed likelihood (right). The transformation changes the prior to a Gaussian and makes the likelihood more complex. The rightmost posterior is smooth and can be written in the form \eqref{eq:lsform}.}\label{fig:2dexp}
\end{figure}

Reference random variables that are \textit{a priori} i.i.d.\ Gaussian can be transformed to each Laplace-distributed element of $D\theta$ using the one-dimensional transformation $\gd$ defined in \eqref{eq:g1D}. Then, $D\theta=g(u)$, where $g: \R^n \rightarrow \R^n$ and
\[g(u) \defeq \big[\gd(u_1), \ldots, \gd(u_n)\big]^\T.\]
Thus, a prior transformation for the $l_1$-type prior is
\begin{equation} \label{eq:gmd}\theta = D^{-1}g(u), \end{equation}
resulting in the requirement that $D$ be invertible. Then, the Jacobian of the transformation is $D^{-1}\J{g}$, where $\J g : \R^n \rightarrow \R^{n\times n}$ is the Jacobian of $g$ given by
\begin{equation}
\label{Jacg} \J {g}(u) = \begin{bmatrix} \gd'(u_{1}) \\
                 &\gd'(u_{2}) &\\
                 & &\ddots &\\
                 & & &\gd'(u_{n})\end{bmatrix},
\end{equation}
and $\gd$ is defined in \eqref{eq:g1D}.

Using this transformation, one can derive the posterior density over $u$ by following the same steps as in the single variable case, with $D^{-1}g(u)$ in place of $\gd(u)$, to obtain
\[
p(u\vert y) \propto  \exp{\left(- \half\Big(f\left(D ^{-1} g(u)\right)-y\Big)^\T  \Gamma_{\textrm{obs}}^{-1} \Big(f\left(D^{-1}g(u)\right)-y\Big) -\half u^\T u \right)}.\]
The transformed posterior is in the form \eqref{eq:lsform} and is amenable to RTO sampling. Figure~\ref{fig:2dexp} illustrates the effect of the transformation on an inverse problem with two unknown parameters, $D = I_2$, and a linear forward model; comparing the second and third columns, we note that the transformed prior becomes a standard Gaussian, while the transformed likelihood becomes non-Gaussian. 

\subsection{RTO-MH with a general prior transformation}\label{sec:rtoprior}

\newcommand{\tT}{\widetilde{T}}
\newcommand{\tF}{\widetilde{F}}
\newcommand{\hu}{u_\text{prop}}
\newcommand{\hui}{u_\text{prop}^{(i)}}
\newcommand{\tf}{\tilde{f}}
\newcommand{\ty}{\tilde{y}}

Given an inverse problem in the form \eqref{eq:model} with a general non-Gaussian prior supported on $\R^n$, suppose that we can construct an invertible and continuously differentiable prior transformation $T: \R^n \rightarrow \R^n$ that couples a standard Gaussian random variable $u$ to our non-Gaussian random variable $\theta$. Both $\gd$ in \eqref{eq:g1D} and $D^{-1} g$ in \eqref{eq:gmd} are examples of such transformations $T$. Then the transformed posterior density is
\begin{align}
\label{eq:TransPost}
p(u\vert y) & \propto  \exp{\left(- \half\Big(f \circ T(u)-y\Big)^\T  \Gamma_{\textrm{obs}}^{-1} \Big(f\circ T(u)-y\Big) -\half u^\T u \right)} \\
& =  \exp{\left(- \half\big\Vert \tf(u)-\tilde{y}\big\Vert ^2  -\half \Vert u\Vert ^2 \right)}\nonumber\\
& \defeq \exp\left(-\frac12\Big\Vert  \tF(u)\Big\Vert ^2\right),\nonumber 
\end{align}
where $\tf(u)=\Gamma_{\textrm{obs}}^{-1/2} f\circ T(u)$ is the transformed forward model, $\tilde{y}=\Gamma_{\textrm{obs}}^{-1/2}y$ is the whitened data, and
$\tF(u)=\left[\begin{array}{c}u\\ \tf(u) - \ty \end{array}\right]$. We can use RTO to sample from the transformed posterior defined by \eqref{eq:TransPost}. 

To perform the optimization steps in RTO and to evaluate the proposal density of RTO, we need the Jacobian of $\tF$, which has the form
\begin{equation}
\label{JacFtilde}
\J {\tF}(u)=\left[\begin{array}{c}I\\ \J {\tf}(u)\end{array}\right].
\end{equation}
Here, $\J {\tf}(u)$ is the Jacobian of the transformed forward model $\tf$ and is given by
\begin{equation}
\label{Jacftilde}
\J{\tf}(u)=\Gamma_{\textrm{obs}}^{-1/2} \J{f} \left( T(u) \right)\J{T}(u),
\end{equation}
where $\J f  : \R^n \rightarrow \R^{m\times n}$ is the Jacobian of the original forward model $f$ and $\J T  : \R^n \rightarrow \R^{n\times n}$ is the Jacobian of the prior transformation $T$. The final algorithm, incorporating a prior transformation in RTO-MH, is summarized in Algorithm~\ref{alg:prior}. 

\begin{algorithm}[htbp]
\caption{RTO-MH with a Prior Transformation} \label{alg:prior}
\begin{algorithmic}[1]
\parState {Determine the prior mapping function $T : \R^n \rightarrow \R^n$ such that $u = T^{-1} (\theta)$ has a standard Gaussian distribution}
\State Find the mode $\overline{u} \in \R^n$ of transformed posterior density $p(u\vert y)$ defined by \eqref{eq:TransPost}
\State Calculate $\overline{Q} \in \R^{n\times m+n}$, whose columns are an orthonormal basis for the column space of $\J {\tF}(\overline{u})$, as defined in \eqref{JacFtilde}--\eqref{Jacftilde}
\For {$i = 1,\ldots, n_{\textrm{samps}}$} in parallel
\State Draw a standard Gaussian sample $\xi^{(i)}\sim\N(0,I_{n})$
\parState {Compute RTO samples via $\hui=\argmin_{u}\big\Vert  \bar{Q}^\T \tF(u)-\xi^{(i)}\big\Vert ^2$ and weights\\ $w\big(\hui\big)=\big|\bar{Q}^\T \J {\tF}\big(\hui\big)\big|^{-1}\exp\left(-\frac12\big\Vert  \tF\big(\hui\big)\big\Vert ^2+\frac12\big\Vert \bar{Q}^\T \tF\big(\hui\big)\big\Vert ^2\right)$}
\EndFor
\For {$i = 1, \ldots, n_\textrm{samps}$} in series
\State Sample $v$ from a uniform distribution on [0,1]
\If {$v < \left. w\big(\hui\big)\right/ w(u^{(i-1)})$}
  \State $u^{(i)}$ = $\hui$
\Else
  \State $u^{(i)}$ = $u^{(i-1)}$
\EndIf
\EndFor
\For {$i = 1,\ldots, n_{\textrm{samps}}$} in parallel
\parState {Define $\theta^{(i)} = T(u^{(i)})$, the desired samples from $p(\theta\vert y)$}
\EndFor
\end{algorithmic}
\end{algorithm}

The computational cost of Algorithm~\ref{alg:prior} is dominated by that of Step 6, where repeated optimization problems are solved and the weights are calculated. Typically, within each optimization iteration, $\tf$ is evaluated once and $J_{\tf}$ is applied to multiple vectors; after optimization, the weight $w(\hui)$ must be evaluated, which requires an evaluation of  $J_{\tf}(\hui)$ and an $\mathcal{O}(n^3)$ computation of the log-determinant.

Under certain conditions on $\tF$, given in Assumption~\ref{assump} (substituting $\tF$ for $F$), the samples $\hu$ generated by Steps 1--7 of Algorithm~\ref{alg:prior} are i.i.d.\ draws from the following probability density: 
\begin{equation}
q(\hu) = (2 \pi)^{-\frac{n}{2}} \left \vert \overline{Q}^\T \J{\widetilde F}(\hu)\right \vert  \exp\left(-\frac12 \Big\Vert  \overline{Q}^\T \widetilde F(\hu) \Big\Vert ^2\right).\label{eq:uproposal}
\end{equation}

When the original forward model is linear, i.e., $f(\theta) = A \theta$, and the prior transformation in Section~\ref{sec:md} is applied, the transformed problem automatically satisfies these conditions. This result is stated in the following theorem and proven in Appendix~\ref{sec:assume}.

\begin{thm}\label{thm:linear}
Let \eqref{eq:l1post} specify the posterior density of a Bayesian inference problem with parameters $\theta \in \R^n$, and let the forward model in  \eqref{eq:l1post}  be linear, $f(\theta) = A \theta$. After the prior transformation \eqref{eq:gmd}, the RTO algorithm described by Steps 1--7 of Algorithm~\ref{alg:prior} generates proposal samples with probability density given in \eqref{eq:uproposal}.
\end{thm}

\noindent
The proof of the theorem simply checks that the transformed problem satisfies the assumptions under which the RTO proposal density holds. For nonlinear forward models $f$, we leave these conditions as an assumption.

\subsection{RTO-MH with an approximate prior transformation} \label{alg:approx}

\newcommand{\hT}{\widehat{T}}
\newcommand{\hF}{\widehat{F}}
\newcommand{\hq}{\hat{q}}
\newcommand{\hp}{\breve{p}}
\newcommand{\hf}{\hat{f}}

The previous section addressed cases where an \textit{exact} prior transformation $T$ is known---i.e., where, if  $\theta$ is distributed according to the prior, then $T^{-1}(\theta)$ has a standard Gaussian distribution. In some cases, determining such an exact transformation might not be feasible. Nonetheless we can still use approximate transformations---that is, transformations which only approximately ``Gaussianize'' the prior---to construct an RTO-MH algorithm. 

Consider a transformation $\hT: \mathbb{R}^n \to \mathbb{R}^n$ that couples a reference random variable $u$ to our prior-distributed random variable $\theta$. But now suppose that the distribution of the reference $u = \hT^{-1}(\theta)$ is only \textit{approximately} Gaussian. (To be clear, the expressions below will not require any Gaussian assumption on $u$; the degree to which $u$ departs from a standard Gaussian will affect the efficiency, not the correctness, of the following Metropolis-Hastings scheme.) These transformations can often be constructed numerically. For example, \cite{tarekTM,mattMCMC,mapsHandbookChapter} describe how to construct parameterized maps from samples or unnormalized density evaluations of any atomless distribution. We can modify our method to work for approximate prior transformations such as these.

As with the exact map, let $\hT$ be invertible and continuously differentiable. We can apply the usual RTO procedure to obtain proposal samples $\hui$ by solving
\[ \hui = \argmin_u \left \Vert \overline{Q}^\T \begin{bmatrix} u \\ \Gamma^{-\frac12}_{\text{obs}} \big( f \circ \hT (u) - y\big)\end{bmatrix}  - \xi^{(i)} \right \Vert^2 \]
for Gaussian samples $\xi^{(i)}$. The proposed samples (in the reference space) will be distributed according to the density
\[\hq(\hu) = (2 \pi)^{-\frac{n}{2}} \left \vert\overline{Q}^\T \J{{\hF}}(\hu)\right \vert \exp\left(-\frac12 \Big \Vert \overline{Q}^\T \hF(\hu) \Big \Vert^2\right) \]
 where
\[\hF(\hu) = \begin{bmatrix} \hu \\ \Gamma^{-\frac12}_{\text{obs}} \big( f \circ \hT (\hu) - y\big)\end{bmatrix}. \]
 In order to obtain samples from the posterior, the RTO-MH algorithm must be modified to incorporate the density of the pullback of the true posterior under the map $\hT$, which has the form
\begin{equation}
\label{e:pullback} 
p(u \vert y) \propto \exp\left(- \frac12 \left \Vert \hf (u)- \tilde{y} \right \Vert^2 \right) \left \vert J_{\hT}(u) \right \vert \,  p_{\theta}  ( \hT(u) ) , 
\end{equation}
where $p_\theta(\cdot )$ is the prior density on $\theta$, $\hf(u) = \Gamma_\text{obs}^{-1} f\circ \hT (u)$, and $\vert J_{\hT}( \cdot) \vert $ is the Jacobian determinant of $\hT$. Contrast \eqref{e:pullback} with \eqref{eq:TransPost}; the key difference is that the standard Gaussian prior on $u$ has been replaced with the pullback of $p_\theta$ under the map $\hT$. If the prior transformation $\hT$ were exact, these two expressions would be equivalent.
This process gives an altered formula for the weights in Step 6 of Algorithm \ref{alg:prior}:
\[ w(\hui) = \left \vert\overline{Q}^\T \J {\hF} (\hui )\right \vert^{-1}\exp\left(-  \frac12 \left \Vert\hf (\hui )- \tilde{y}  \right \Vert^2 +\frac12\left \Vert\overline{Q}^\T \hF (\hui )\right \Vert^2\right)   \left \vert J_{\hT} (\hui ) \right \vert \, \, p_{\theta} ( \hT (\hui ) ).  \] 
The rest of the algorithm remains unchanged. In essence, the error in the approximate prior transformation is handled by appropriately altering the Metropolis-Hastings acceptance ratio.

\resetcounters

\section{Numerical examples} \label{sec:ex}
We apply RTO-MH with prior transformations to three numerical examples, labeled A, B, and C, all with $l_1$-type priors. Examples A and B are (spatially) 1-D deconvolution problems with linear forward models, while Example C is a (spatially) 2-D inverse problem with a nonlinear forward model. In Example A, we use a TV prior and perform a simple comparison of the efficiency of our method with that of other MCMC samplers, including the Gibbs scheme proposed in \cite{felix} for linear inverse problems with $l_1$-type priors. In Example B, we use a Besov \B space prior and examine the effects of parameter dimension $n$ and hyperparameter $\lambda$ on the performance of RTO. Finally, in Example C, we infer the coefficient field of a linear elliptic PDE; in this case, we use the 2-D Besov \B space prior. This example is meant to test RTO on a more difficult inverse problem, involving a nonlinear forward model and a parameter field in two spatial dimensions.

\subsection{One-dimensional deconvolution problems}
Examples A and B involve the deconvolution of a 1-D signal. We discretize a true signal, $\theta_\text{true}(x)$, defined on the domain $x \in [0,1]$, using $n$ grid points. The true signal is convolved with the function 
\begin{equation}\label{eq:k}
k(x) = \left\{\begin{matrix} 1 &\text{  if  } -\frac{1}{64}<x<\frac{1}{64}\\
 0 &\text{  otherwise,  }\end{matrix}\right.
\end{equation}
and evaluated at $m = 30$ points to create measurements corresponding to integrals over interior segments of the domain. The data $y \in \R^m$ are generated by adding i.i.d.\ Gaussian noise with $\Gamma_\text{obs} = \sigma_\text{obs}^2I$.

\subsubsection{Example A: TV prior}
In this example, the true signal is the square pulse,
\[\theta_\text{true}(x) = \left\{\begin{matrix} 1 &\text{  if  } \frac{1}{3}<x<\frac{2}{3}\\
 0 &\text{  otherwise  }\end{matrix}\right.,\]
 which is also used in \cite{tvnotok,felix}. Figure~\ref{fig:tv_expost} depicts the true signal and the resulting data.
\begin{figure}[htbp]
\centering
\pic{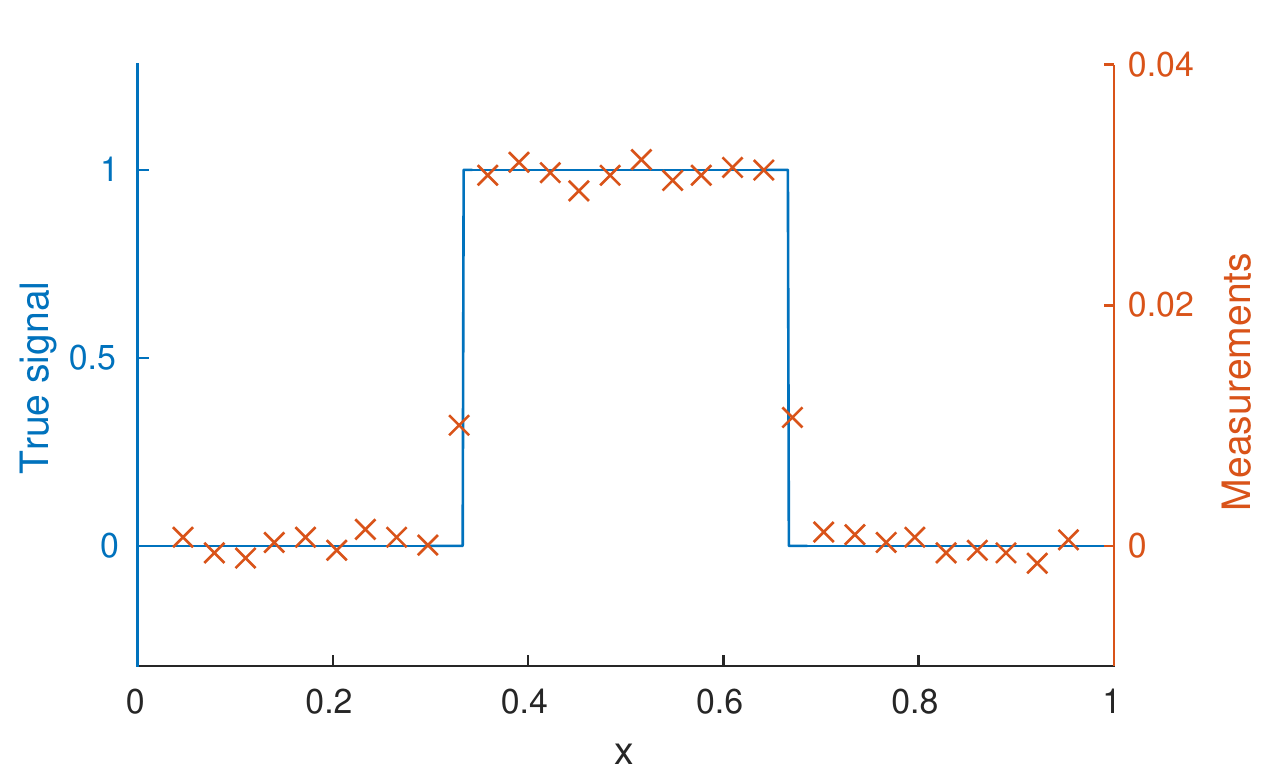}{0.7}
\caption{The true signal and noisy measurements for example A.}\label{fig:tv_expost}
\end{figure}

We use a TV prior, i.e.,
$\pi(\theta) \propto \exp{\left( - \lambda \| D \theta \|_1 \right)},$
with $\theta \in \R^{n}$ , $n = 63$,
\[ D = \begin{bmatrix}
1 &  &  & 1  \\
-1 & 1 &  &  \\
 & \ddots & \ddots & \\
 &  & -1 & 1
\end{bmatrix}_{n\times n}, \]
$\sigma_\text{obs} = \epow{1}{-3}$, and $\lambda = 8$. The first row of $D$ imposes the condition that the sum of the boundary values is zero, making $D$ invertible, which is required for the prior transformation to be well-defined. 

We generate MCMC chains using three different algorithms: RTO-MH with a prior transformation,  MALA,  and the Gibbs scheme of \cite{felix}. To compare computational costs, we count the number of function evaluations used by each algorithm, with a Jacobian evaluation (used by MALA and RTO) counted as a single function evaluation for this linear problem. For each algorithm, we stopped the MCMC chain once the number of evaluations reached $\epow{4}{6}$. 
The resulting MCMC chains are shown in Figure~\ref{fig:tv_chain}. For RTO, we used the default settings of the nonlinear least-squares solver \texttt{lsqnonlin} in MATLAB to perform all the optimizations. Our first attempt at MALA used the adaptive (AMALA) scheme of \cite{amala}. The resulting chain did not reach stationarity after $\epow{4}{6}$ evaluations, as seen in Figure~\ref{fig:tv_chain}. Note that the vertical axis of the figure showing the AMALA chain is different from the others; the chain has not even located the region of high posterior probability. Instead, to obtain a convergent solution using MALA, we switched to a preconditioned MALA scheme, where the preconditioner was prescribed to be the posterior covariance matrix estimated from a converged MCMC chain generated by another algorithm (e.g., Gibbs sampling). Since finding this covariance requires a full exploration of the posterior, this scheme is not something that could be applied in practice; rather, it represents the ``ideal'' or endpoint of any AMALA scheme. But we show these MALA results here simply for comparative purposes. As seen in Figures~\ref{fig:tv_mean} and \ref{fig:tv_cov}, the posterior mean (also called the conditional mean (CM)) and posterior covariance from all three MCMC algorithms agree as we increase the maximum number of evaluations. This provides numerical evidence that RTO-MH with a prior transformation generates samples from the correct distribution. 

\begin{figure}[htbp]
\centering
\subfloat[RTO with prior transformation]{\pic{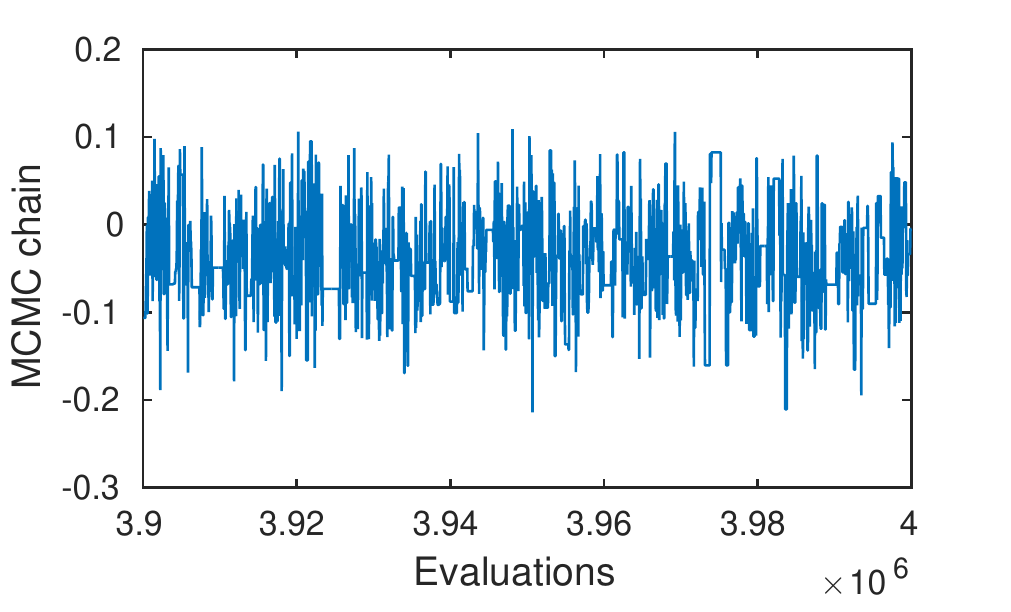}{0.4975}}
\subfloat[AMALA]{\pic{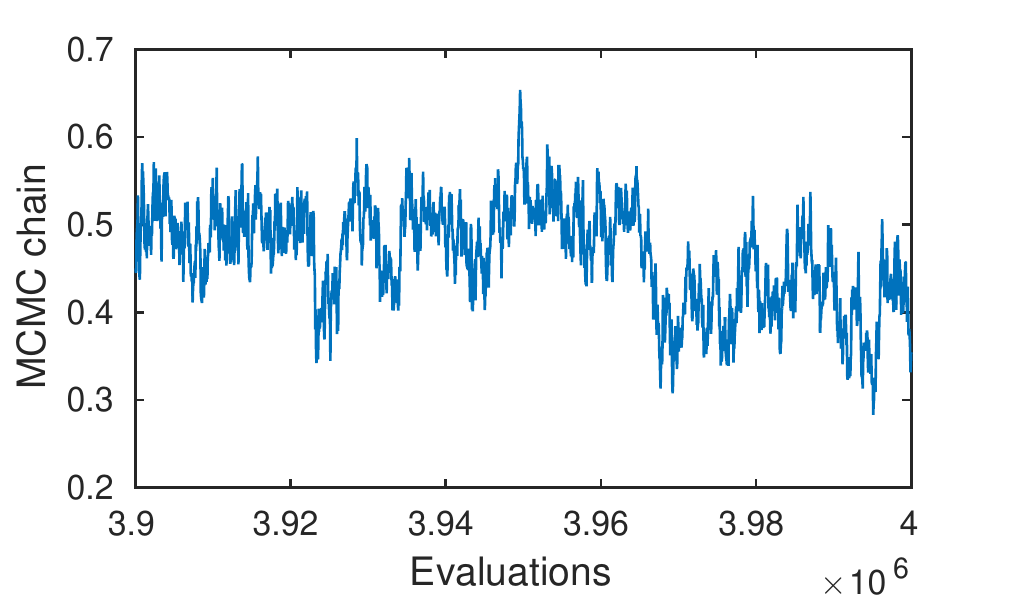}{0.4975}}\\
\subfloat[MALA with ideal preconditioner]{\pic{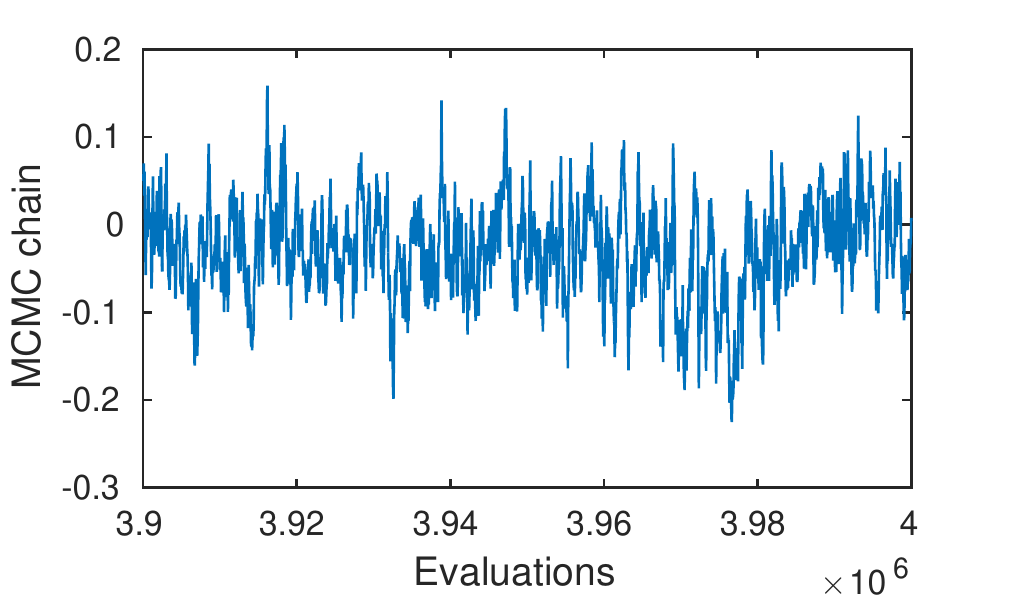}{0.4975}}
\subfloat[Gibbs]{\pic{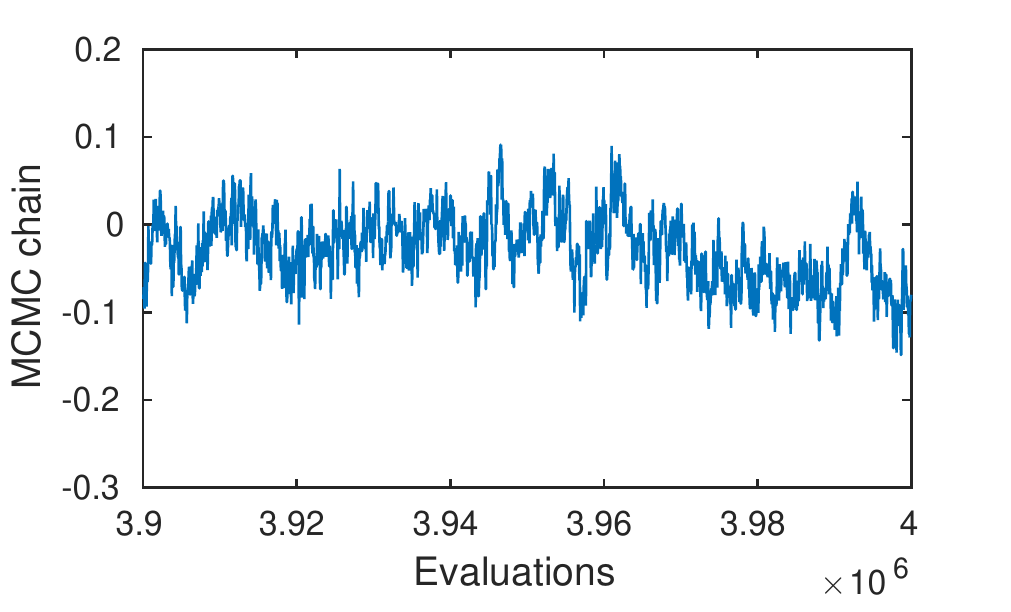}{0.4975}}
\caption{Example A: MCMC chains from various methods. Index $7$ is plotted, which corresponds to the median ESS for RTO. The chain for AMALA is not yet stationary. The horizontal axis (number of function and Jacobian evaluations) reflects a common measure of computational cost for all methods.}\label{fig:tv_chain}
\end{figure}

\begin{figure}[htbp]
\centering
\subfloat[Budget of $\epow{1}{4}$ evaluations.]{\pic{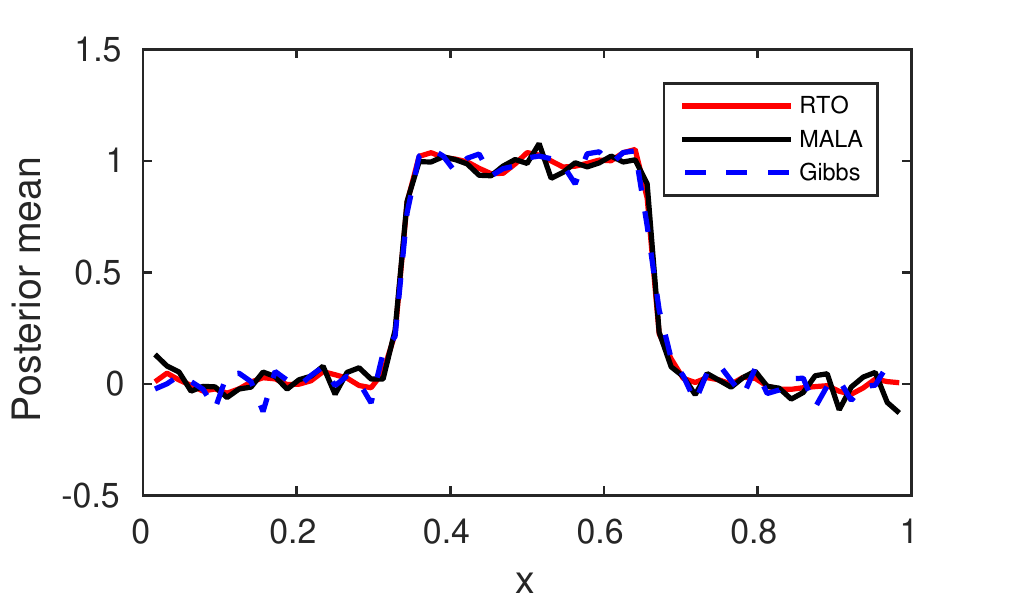}{0.4975}}
\subfloat[Budget of $\epow{4}{6}$ evaluations.]{\pic{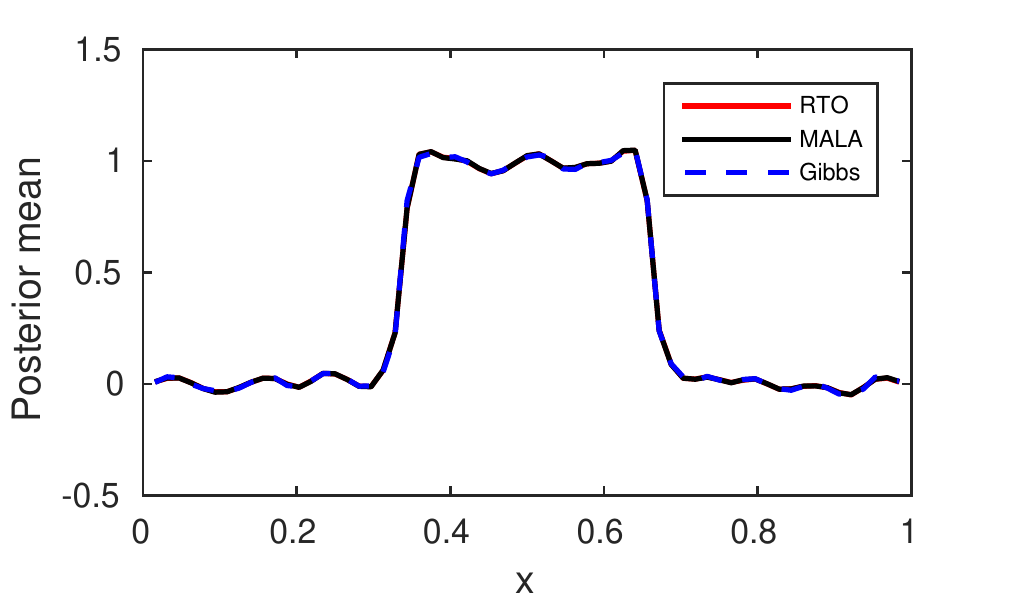}{0.4975}}
\caption{Example A: Sample estimates of the posterior mean, computed with transformed RTO (red), MALA (black), and Gibbs (blue).}\label{fig:tv_mean}
\end{figure}

\begin{figure}[htbp]
\centering
\subfloat[Estimated posterior covariance using RTO]{\pic{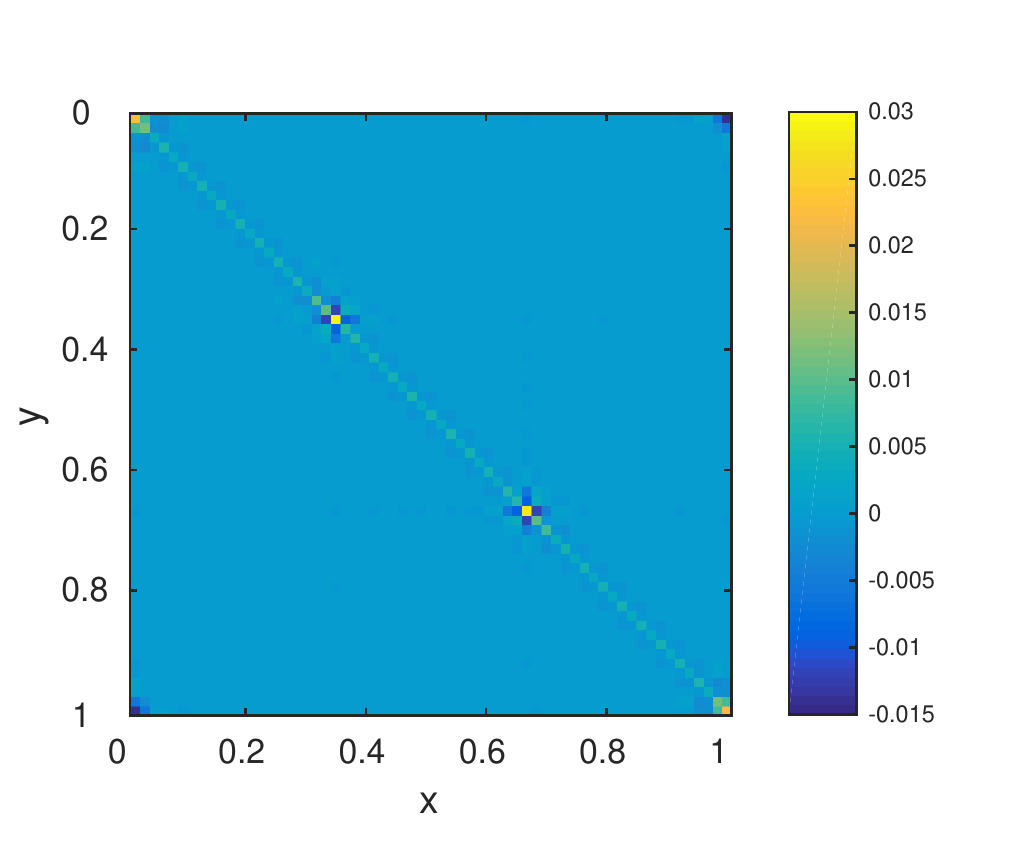}{0.4975}}
\subfloat[Estimated posterior covariance using MALA]{\pic{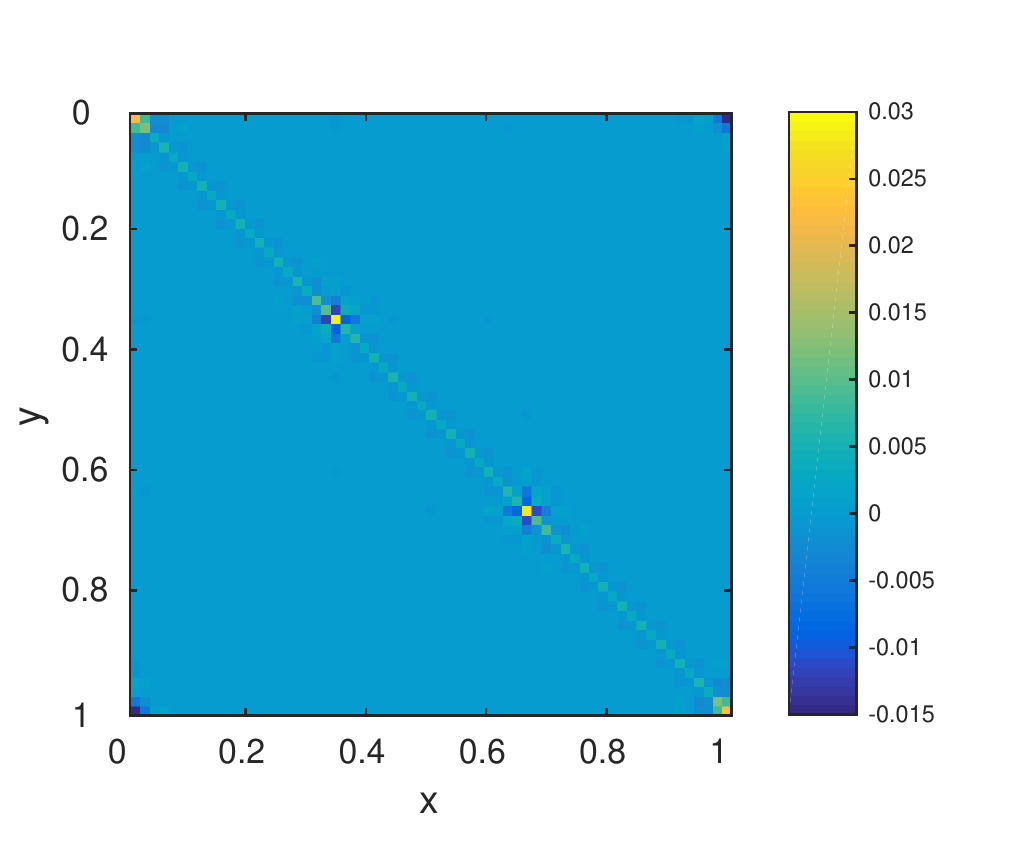}{0.4975}}\\
\subfloat[Estimated posterior covariance using Gibbs]{\pic{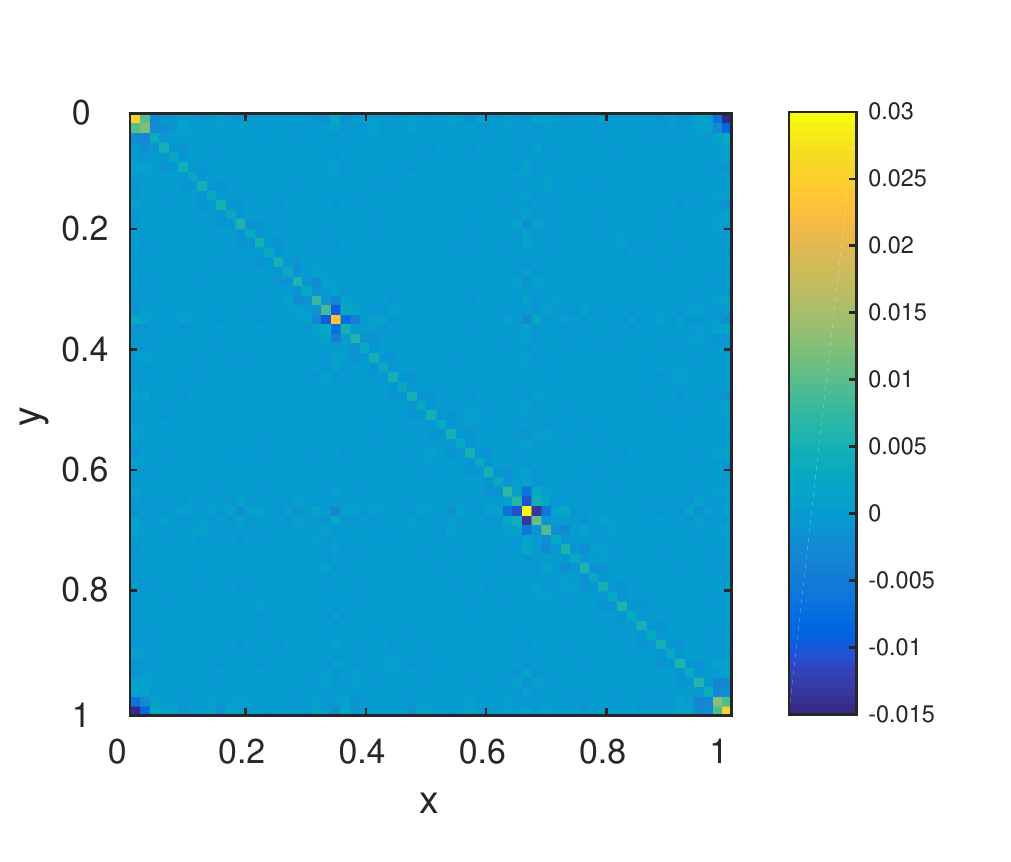}{0.4975}}
\subfloat[Estimates of the posterior standard deviation]{\pic{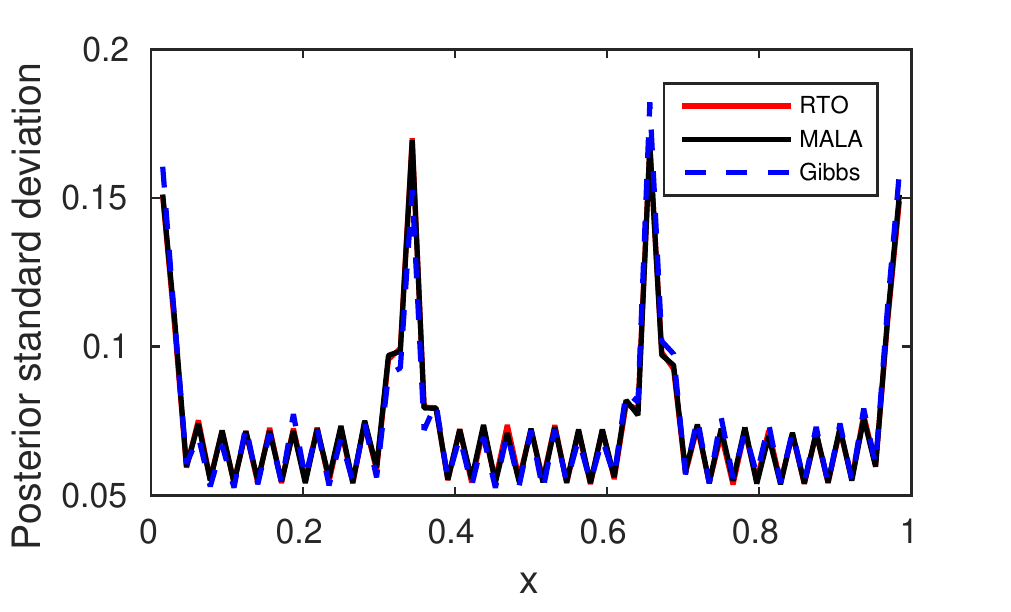}{0.475}}
\caption{Example A: Sample estimates of the posterior covariance (top row and bottom left) and pointwise posterior standard deviation (bottom right), using a budget of $\epow{4}{6}$ evaluations.}\label{fig:tv_cov}
\end{figure}

Next, we assess effective sample size (ESS) per function/Jacobian evaluation and per CPU-second, as two measures of computational efficiency. ESS is the number of effectively independent samples in a Markov chain, i.e., the number of 
samples in a standard Monte Carlo estimator that has the same variance as an estimator computed from the correlated samples of the MCMC chain. It can be interpreted as a measure of the quality of the MCMC samples, where larger values of ESS indicate better chain mixing \cite{mcmcinpractice}. An accurate way to calculate the ESS of an MCMC chain of a single parameter is found in \cite{wolff}; we do so for each component of our chains and report the minimum, median, and maximum (across components) ESS per evaluation and ESS per CPU-second in Table~\ref{tab:tv_perESS}. The RTO method has a higher ESS per evaluation than the other benchmark algorithms. 
However,  since the optimization and calculation of the weights in RTO involves additional computational overhead, MALA using the ``ideal'' preconditioner has a higher ESS per CPU-second than RTO-MH. As noted above, though, MALA with the ``ideal'' preconditioner is not a practically realizable algorithm. AMALA is a practical realization of preconditioned MALA, and the chain's poor mixing is reflected in low ESS per CPU-second values.
Overall, these results suggest that RTO-MH with a prior transformation is quite competitive for this test case, even without accounting for the fact that RTO can be run in parallel. 
\begin{table}[htbp]
\caption{Example A: ESS per evaluation or per CPU-second. Each Jacobian evaluation is considered to be equivalent in cost to one function evaluation. MALA (ideal) is preconditioned with the posterior covariance calculated from a converged chain of another method.} \label{tab:tv_perESS}
\centering
\begin{tabular}{l l l l l l l}
\toprule
\multirow{2}{*}{Method} & \multicolumn{3}{c}{ESS per evaluation} & \multicolumn{3}{c}{ESS per CPU-second}\\ 
\cmidrule(l{2pt}r{2pt}){2-4} \cmidrule(l{2pt}r{2pt}){5-7} 
 &Minimum &Median &Maximum &Minimum &Median &Maximum \\ 
 \midrule
RTO with transf. & $\mepow{2.48}{-3}$ & $\mepow{7.43}{-3}$ & $\mepow{8.72}{-3}$ & $\mepow{4.77}{-1}$ & $\mepow{1.43}{0}$ & $\mepow{1.67}{0}$ \\
AMALA & $\mepow{1.09}{-6}$ & $\mepow{1.21}{-6}$ & $\mepow{3.76}{-6}$ & $\mepow{3.19}{-4}$ & $\mepow{3.54}{-4}$ & $\mepow{1.10}{-3}$ \\
MALA (ideal) & $\mepow{1.08}{-3}$ & $\mepow{1.24}{-3}$ & $\mepow{1.48}{-3}$ & $\mepow{1.39}{1}$ & $\mepow{1.60}{1}$ & $\mepow{1.90}{1}$ \\
Gibbs               & $\mepow{9.60}{-6}$ & $\mepow{7.06}{-5}$ & $\mepow{1.26}{-4}$ & $\mepow{1.96}{-2}$ & $\mepow{1.44}{-1}$ & $\mepow{2.57}{-1}$ \\
\bottomrule
\end{tabular}
\end{table}

\begin{rem}
The CM estimate using a TV prior is in general not piecewise constant (i.e., blocky). In \cite{tvnotok}, it is proven that under refinement of parameter discretization, the CM estimate using a TV prior will become smooth.
\end{rem}

\subsubsection{Example B: Besov space prior}
This second example is also a deconvolution of a 1-D signal. Here, the true signal is taken to be
\[\theta_\text{true}(x) = \left\{\begin{matrix} 1 &\text{  if  } \left.{2}\right/{15}<x<\left.{7}\right/{15}\\
 \frac{1}{2} &\text{  if  }\left.{10}\right/{15}<x<\left.{13}\right/{15}\\
 0 &\text{  otherwise  }\end{matrix}\right..\]
 Figure~\ref{fig:bs_expost} shows the true signal and resulting data.

\begin{figure}[htbp]
\centering
\pic{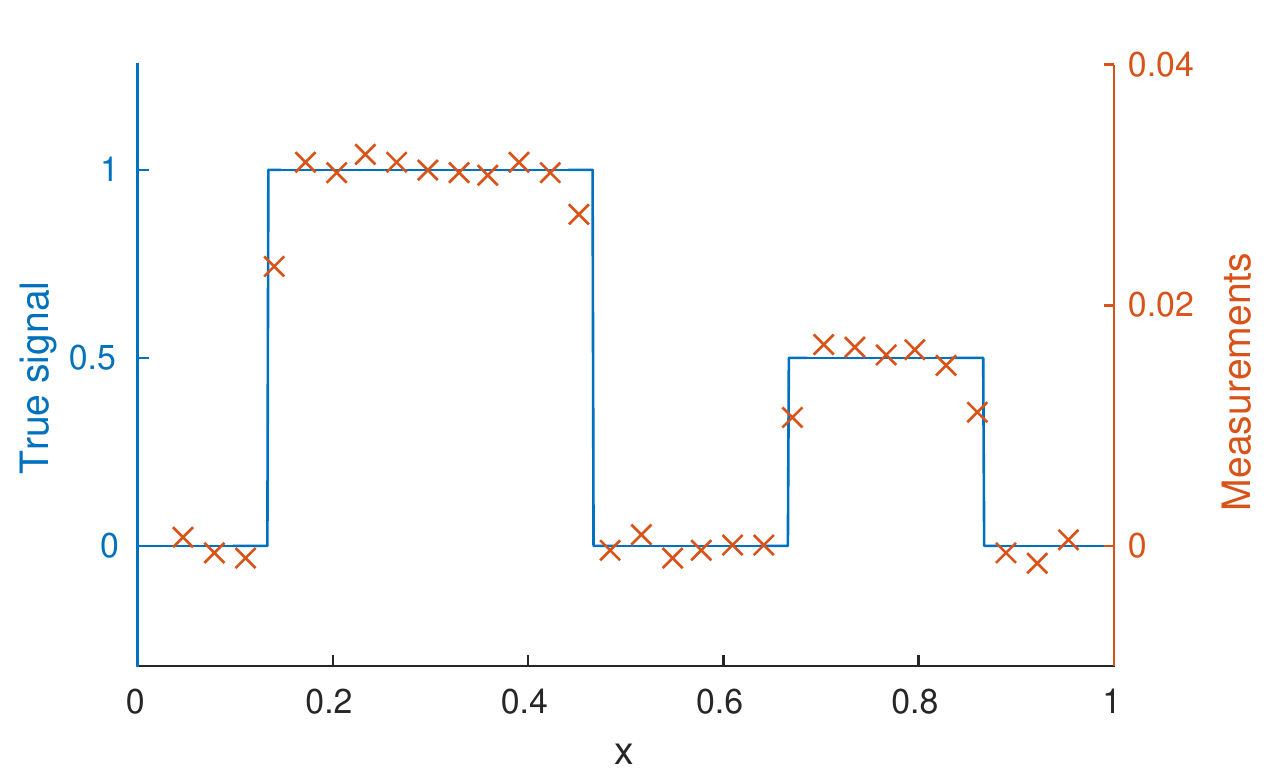}{0.7}
\caption{Example B: True signal and noisy measurements.}\label{fig:bs_expost}
\end{figure}

This time, we use the Besov \B prior with $s=1$ and Haar wavelets, so that again
$\pi(\theta) \propto \exp{\left( - \lambda \| D \theta \|_1 \right)},$
where $\theta \in \R^{n}$, $D \in \R^{n\times n},$ and $\lambda \in \R$. In this case, the matrix $D$ contains scaled wavelet basis functions (see details in Appendix~\ref{sec:besov}), and $n$ must be a power of $2$. We set the observational noise to be $\sigma_\text{obs} = \epow{1}{-3}$.

RTO with a prior transformation is used to sample from the posterior distributions. We perform two studies: first by fixing the hyperparameter to $\lambda = 32$ and scanning through parameter dimensions $n \in \{32, 64, 128, 258, 512\}$; and second, by fixing $n = 64$ and scanning through hyperparameter values $\lambda \in \{\frac12, 1, 2, 4, 8, 16, 32, 64, 128\}$. We use chain lengths of $\epow{1}{4}$, and tabulate the total ESS and the number of function and Jacobian evaluations. 
When we increase the dimension $n$, the posterior mean converges, as in Figure~\ref{fig:bs_mean}. This is expected due to the discretization-invariant nature of the Besov \B prior \cite{discretization-invariant,besov}. 
Next, as reported in Table~\ref{tab:bs_ESS}, with each doubling of the dimension $n$, the ESS does not really decrease and the number of function evaluations increases only slightly. This is an important and encouraging result, as it is evidence of discretization invariance not only in the problem formulation, but in the performance of the transformed RTO-MH \textit{sampling scheme}.
Finally, as we increase the hyperparameter $\lambda$, the CM becomes smoother and the posterior standard deviation decreases, as in shown Figure~\ref{fig:bs_l_mean}.  The sampling efficiency of our algorithm also deteriorates with increasing $\lambda$, as shown in Table~\ref{tab:bs_l_ESS}. Overall, the results from these parameter studies indicate that RTO-MH with a prior transformation is effective even when the parameter dimension $n$ is in the hundreds.

\begin{rem}
In Figure~\ref{fig:bs_mean}, the posterior standard deviation does not converge as the discretization is refined (i.e., as $n$ increases). This behavior is not unexpected, as the prior standard deviation also does not converge under mesh refinement. In particular, the \B Besov space prior with Haar wavelets has finite pointwise variance only when $s>1$, and not when $s=1$. One can prove this property by summing the variance contributions from each level of wavelets in the Besov prior, as shown in Appendix~\ref{sec:priorvar}.
\end{rem}

\begin{rem}
One possible reason for the decrease in sampling efficiency with higher $\lambda$ is that the posterior samples lie further in the tails of the Laplace prior. As a result, the transformation is more nonlinear in the sense that the Hessian involving $g_{1D}''$ is of higher magnitude.
\end{rem}

\begin{figure}[htbp]
\centering
\subfloat[Posterior mean]{\pic{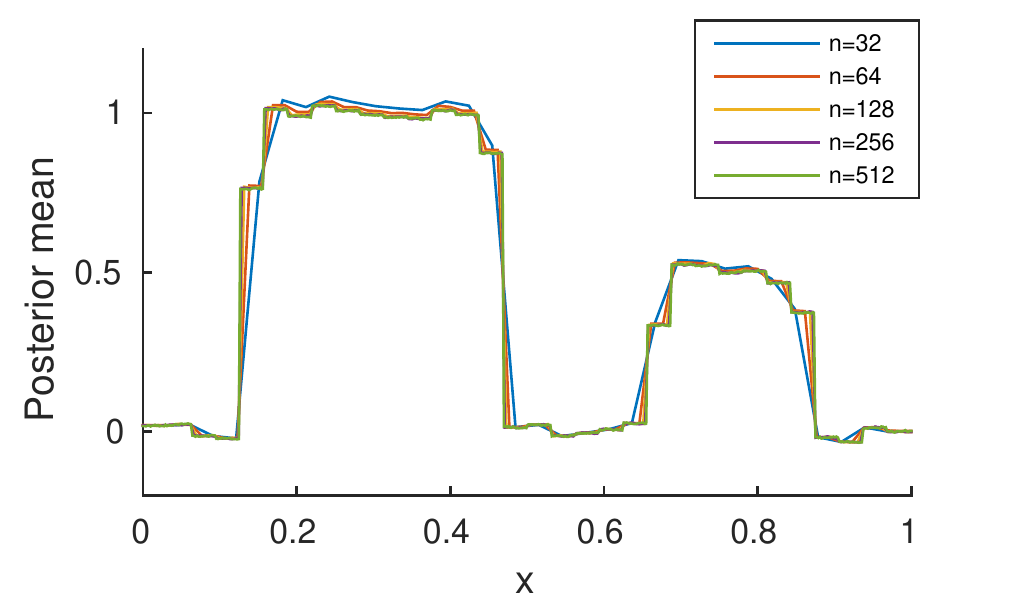}{0.7}}\\
\subfloat[Posterior standard deviation]{\pic{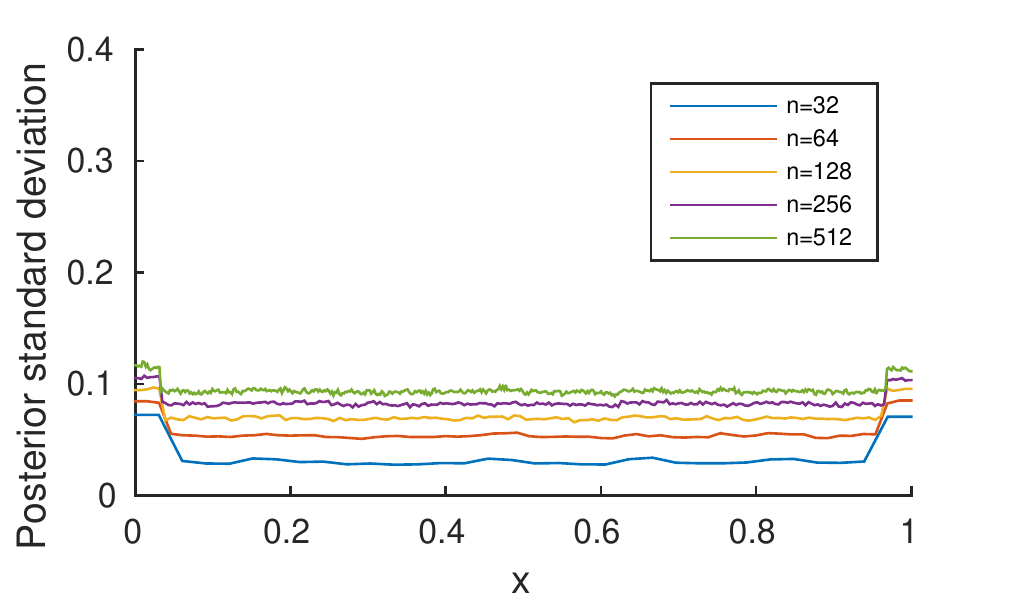}{0.7}}
\caption{Example B: Posterior mean and standard deviation for different values of the parameter dimension $n$. Hyperparameter $\lambda$ is fixed to $32$.}\label{fig:bs_mean}
\end{figure}

\begin{table}[htbp]
\centering
\caption{Example B: ESS and computational cost of RTO for various parameter dimensions, given chains of length $\epow{1}{4}$.} \label{tab:bs_ESS}
\begin{tabular}{l l l l l l}
\toprule
\multirow{2}{*}{$n$} & \multicolumn{3}{c}{Total ESS} & \multicolumn{2}{c}{Total evaluations}\\ \cmidrule(r){2-4} \cmidrule{5-6}
 &Minimum &Median &Maximum &Function &Jacobian\\ 
 \midrule
  $32$  & $\mepow{2.68}{3}$   & $\mepow{3.86}{3}$   & $\mepow{4.61}{3}$  & $\mepow{4.26}{5}$ & $\mepow{4.26}{5}$\\
  $64$  & $\mepow{2.63}{3}$   & $\mepow{3.65}{3}$   & $\mepow{4.44}{3}$  & $\mepow{4.55}{5}$ & $\mepow{4.55}{5}$\\
  $128$ & $\mepow{2.10}{3}$   & $\mepow{3.53}{3}$   & $\mepow{5.07}{3}$  & $\mepow{4.59}{5}$ & $\mepow{4.59}{5}$\\
  $256$ & $\mepow{2.89}{3}$   & $\mepow{3.69}{3}$   & $\mepow{4.43}{3}$  & $\mepow{4.61}{5}$ & $\mepow{4.61}{5}$\\
  $512$ & $\mepow{2.06}{3}$   & $\mepow{3.65}{3}$   & $\mepow{4.41}{3}$  & $\mepow{4.65}{5}$ & $\mepow{4.65}{5}$\\ 
\bottomrule
\end{tabular}
\end{table}

\begin{figure}[htbp]
\centering
\subfloat[Posterior mean]{\pic{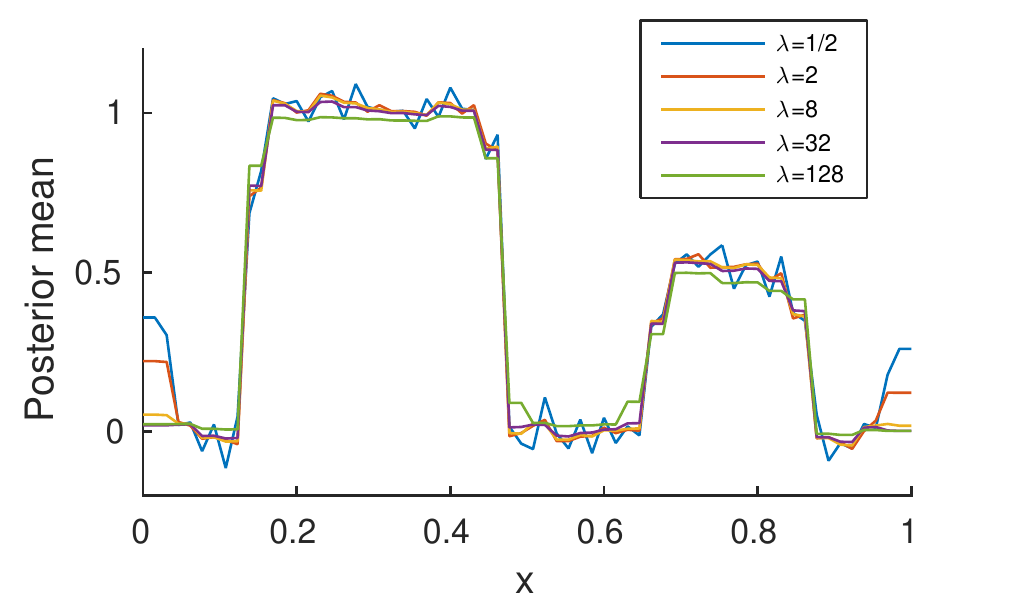}{0.7}}\\
\subfloat[Posterior standard deviation]{\pic{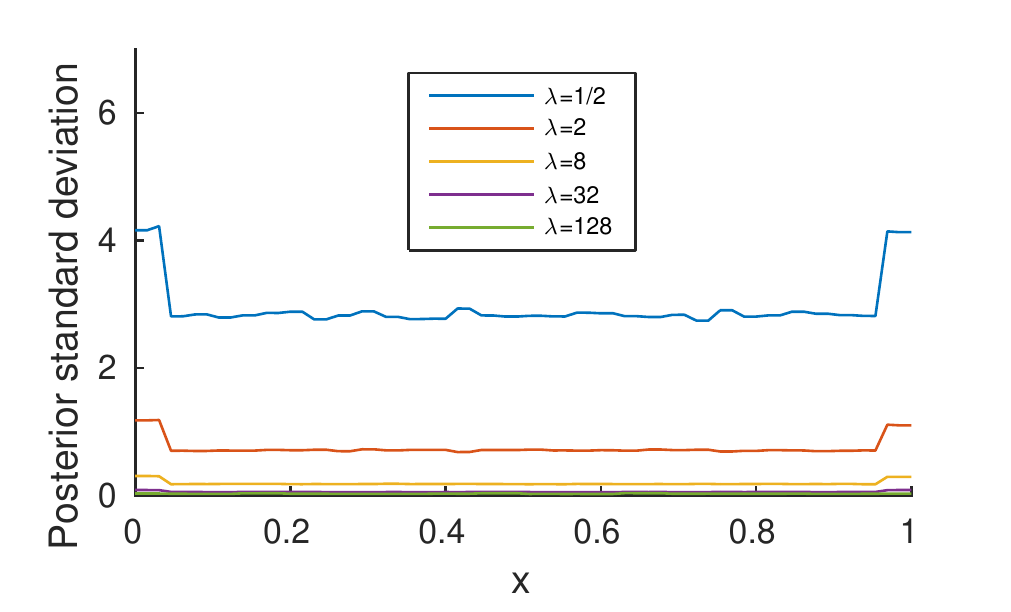}{0.7}}
\caption{Example B: Posterior mean and standard deviation for different values of the hyperparameter $\lambda$. Parameter dimension $n$ is fixed to $64$.}\label{fig:bs_l_mean}
\end{figure}

\begin{table}[htbp]
\centering
\caption{Example B: ESS and computational cost of RTO for varying prior hyperparameter values $\lambda$. Chains of length $\epow{1}{4}$ are used.} \label{tab:bs_l_ESS}
\begin{tabular}{l l l l l l}
\toprule
\multirow{2}{*}{$\lambda$} & \multicolumn{3}{c}{Total ESS} & \multicolumn{2}{c}{Total evaluations}\\ 
\cmidrule(r){2-4} \cmidrule{5-6}
 &Minimum &Median &Maximum &Function &Jacobian \\ 
 \midrule
   $0.5$ & $\mepow{5.00}{3}$ & $\mepow{5.83}{3}$ & $\mepow{7.88}{3}$ & $\mepow{5.47}{5}$ & $\mepow{5.47}{5}$\\
     $1$ & $\mepow{5.66}{3}$ & $\mepow{6.30}{3}$ & $\mepow{8.11}{3}$ & $\mepow{5.05}{5}$ & $\mepow{5.05}{5}$\\
     $2$ & $\mepow{5.74}{3}$ & $\mepow{6.71}{3}$ & $\mepow{8.23}{3}$ & $\mepow{4.73}{5}$ & $\mepow{4.73}{5}$\\
     $4$ & $\mepow{5.82}{3}$ & $\mepow{6.51}{3}$ & $\mepow{8.01}{3}$ & $\mepow{4.63}{5}$ & $\mepow{4.63}{5}$\\
     $8$ & $\mepow{4.68}{3}$ & $\mepow{5.69}{3}$ & $\mepow{6.96}{3}$ & $\mepow{4.69}{5}$ & $\mepow{4.69}{5}$\\
    $16$ & $\mepow{3.20}{3}$ & $\mepow{4.39}{3}$ & $\mepow{5.29}{3}$ & $\mepow{4.77}{5}$ & $\mepow{4.77}{5}$\\
    $32$ & $\mepow{2.63}{3}$ & $\mepow{3.65}{3}$ & $\mepow{4.44}{3}$ & $\mepow{4.55}{5}$ & $\mepow{4.55}{5}$\\
    $64$ & $\mepow{2.32}{3}$ & $\mepow{3.55}{3}$ & $\mepow{4.34}{3}$ & $\mepow{3.83}{5}$ & $\mepow{3.83}{5}$\\
   $128$ & $\mepow{1.08}{3}$ & $\mepow{2.19}{3}$ & $\mepow{2.79}{3}$ & $\mepow{3.02}{5}$ & $\mepow{3.02}{5}$\\
\bottomrule
\end{tabular}
\end{table}

\subsection{Two-dimensional elliptic PDE inverse problem}
Our next numerical example is an elliptic PDE coefficient inverse problem on a two-dimensional domain. The forward model maps the log-conductivity field of the Poisson equation to observations of the potential field,
\[ - \nabla \cdot \Big(\exp\big(\theta(x)\big) \nabla s(x) \Big) = h(x), \ \ x \in [0,1]^2 ,\]
where $\theta$ is the log-conductivity, $s$ is the potential, and $h$ is the forcing function. Neumann boundary conditions
\[ \exp\big(\theta(x)\big) \nabla s(x) \cdot \vec n(x) = 0\]
are imposed, where $\vec n(x)$ is the normal vector at the boundary. To complete the system of equations, the average potential on the boundary is set to zero.

This PDE is solved using finite elements. The domain is partitioned into a $\sqrt n\times \sqrt n$ uniform grid of square elements, and we use linear shape functions in both directions. The parameters $\theta \in \R^n$ to be inferred are the nodal values of $\theta(x)$. Independent Gaussian noise with standard deviation $\sigma_\text{obs} = \epow{2}{-3}$ is added to the potential field $s$ to give the observational data $y$. 

The forcing function $h$ is a linear combination of thirteen Gaussian bumps: nine with weight~$1$ centered at the points $(a,b)$, $a,b\in \{0.05, 0.5, 0.95\}$, and four with weight~$-9/4$ centered at the points $(a,b)$, $a,b\in\{0.25, 0.75\}$. The true parameter field $\theta_{\text{true}}$, forcing function, and resulting noisy measurements are shown in Figure~\ref{fig:2d_expost}. A similar problem setup is found in \cite{dili}.

\begin{figure}[htbp]
\centering
\subfloat[True signal]{\pic{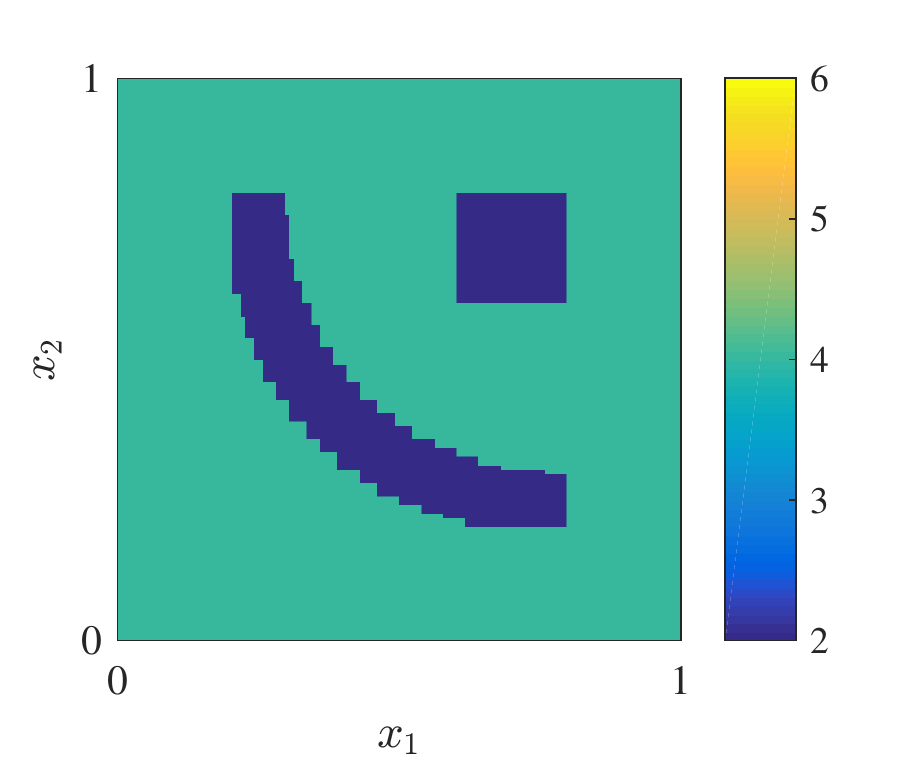}{0.33}}
\subfloat[Forcing function]{\pic{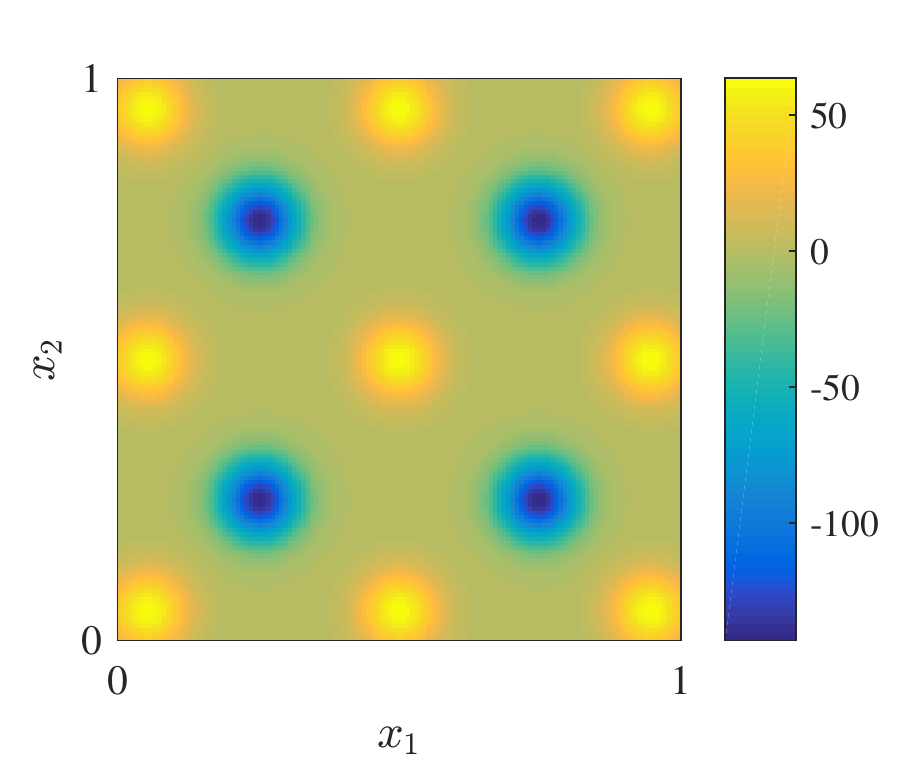}{0.33}}
\subfloat[Noisy measurements]{\pic{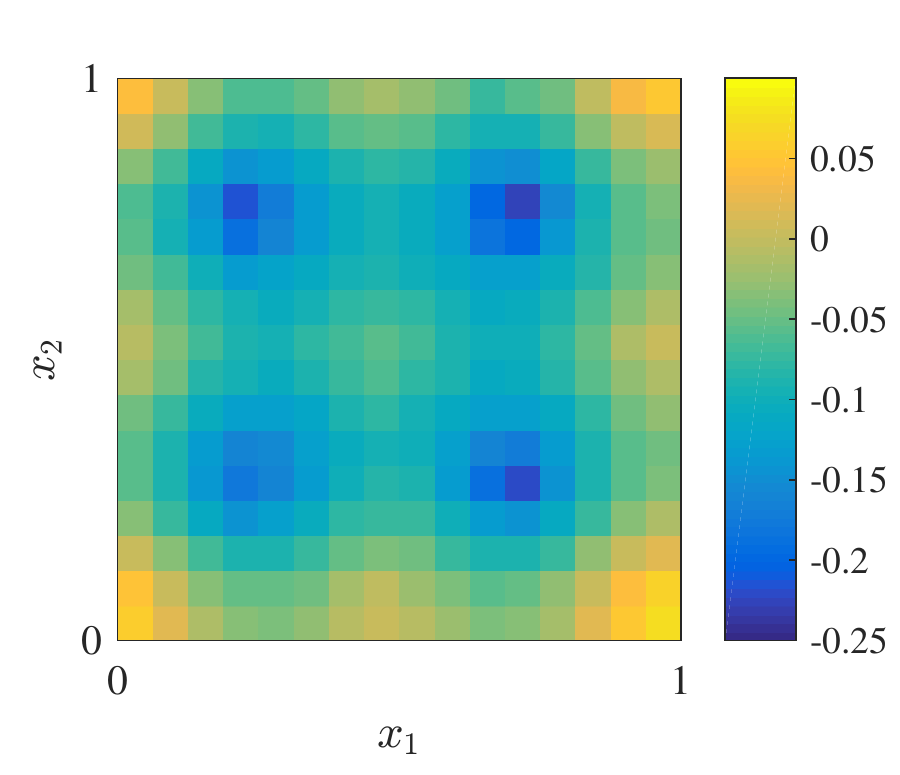}{0.33}}
\caption{Example C: True signal, forcing function, and noisy measurements.}\label{fig:2d_expost}
\end{figure}

\subsubsection{Example C: Besov space prior in 2-D}

To complete the setup of the Bayesian inverse problem, we impose a 2-D Besov \B prior, with a tensorized Haar wavelet basis, on $\theta$. This $l_1$-type prior is also written in the form \eqref{eq:l1prior}. The columns of matrix $D$ are Kronecker products of the columns of the matrix from the 1-D Besov \B space prior. The hyperparameter value is $\lambda = 32$ and the parameter dimension is set to $n=256$, which gives rise to a $16\times16$ grid. The observational data are generated using a finer $128 \times 128$ grid.

We ran RTO-MH with a prior transformation and generated an MCMC chain of length $\epow{2}{5}$. The computation used $\epow{9.3}{6}$ function evaluations and $\epow{9.3}{6}$ Jacobian evaluations to produce an ESS of $\epow{4.5}{2}$. The posterior mean, estimated from the MCMC samples, appears similar to $\theta_\text{true}$ as shown in Figure~\ref{fig:2d_mean}. We also estimate the posterior standard deviation, shown in Figure~\ref{fig:2d_mean}; lower uncertainty regions seem to coincide with smaller log-conductivities. It is also instructive to look at independent samples from the posterior, shown in Figure~\ref{fig:2d_samp}(a). They exhibit small-scale roughness (consistent with the Besov prior) and sample-to-sample variability; however, the corresponding samples from the posterior predictive distribution closely match the data, as in Figure~\ref{fig:2d_samp}(b). The fact that wider variability among the parameter samples corresponds to much narrower variability among the predictions reflects the smoothing properties of the forward operator and the ill-posedness of the inverse problem. Collectively the posterior samples $\{ \theta^{(i)} \}$ characterize uncertainty in the solution of the inverse problem.

We note that the Gibbs sampler of \cite{felix} does not extend to nonlinear inverse problems such as this test case. 

\begin{figure}[htbp]
\centering
\subfloat[Posterior mean]{\pic{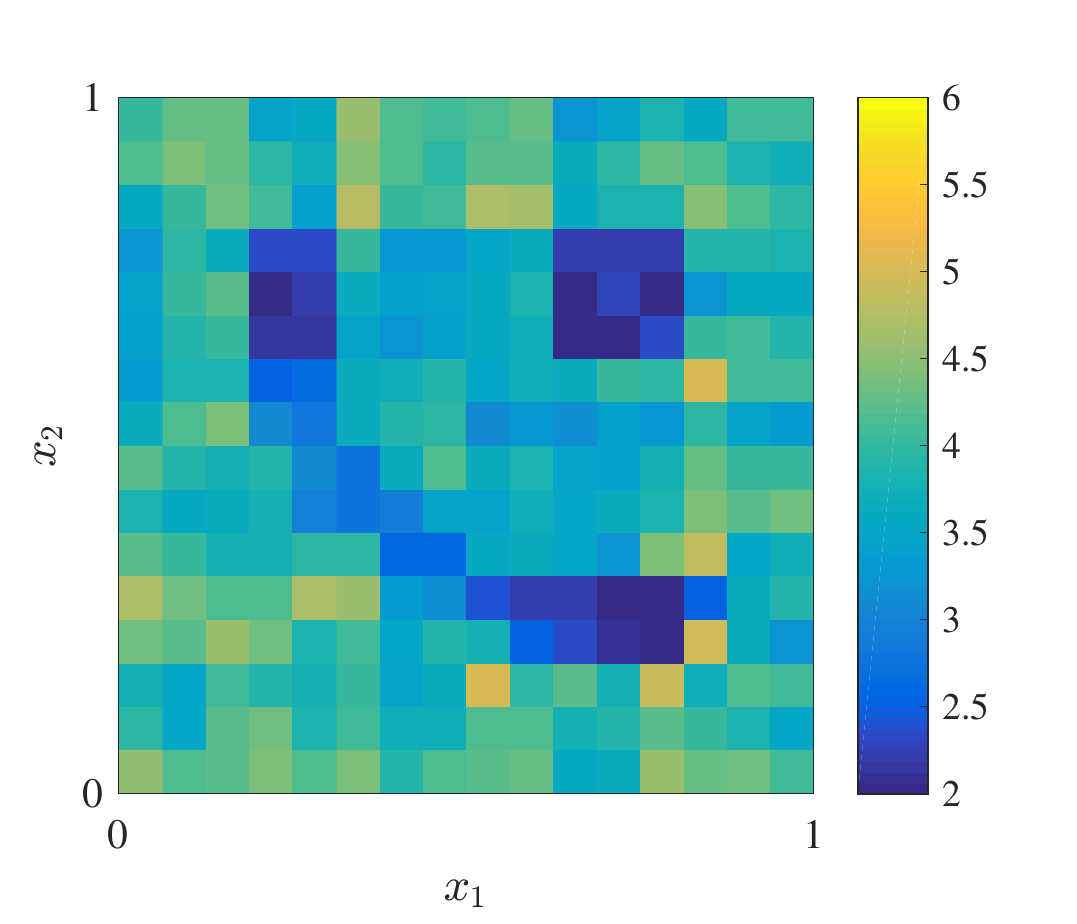}{0.49}}
\subfloat[Posterior standard deviation]{\pic{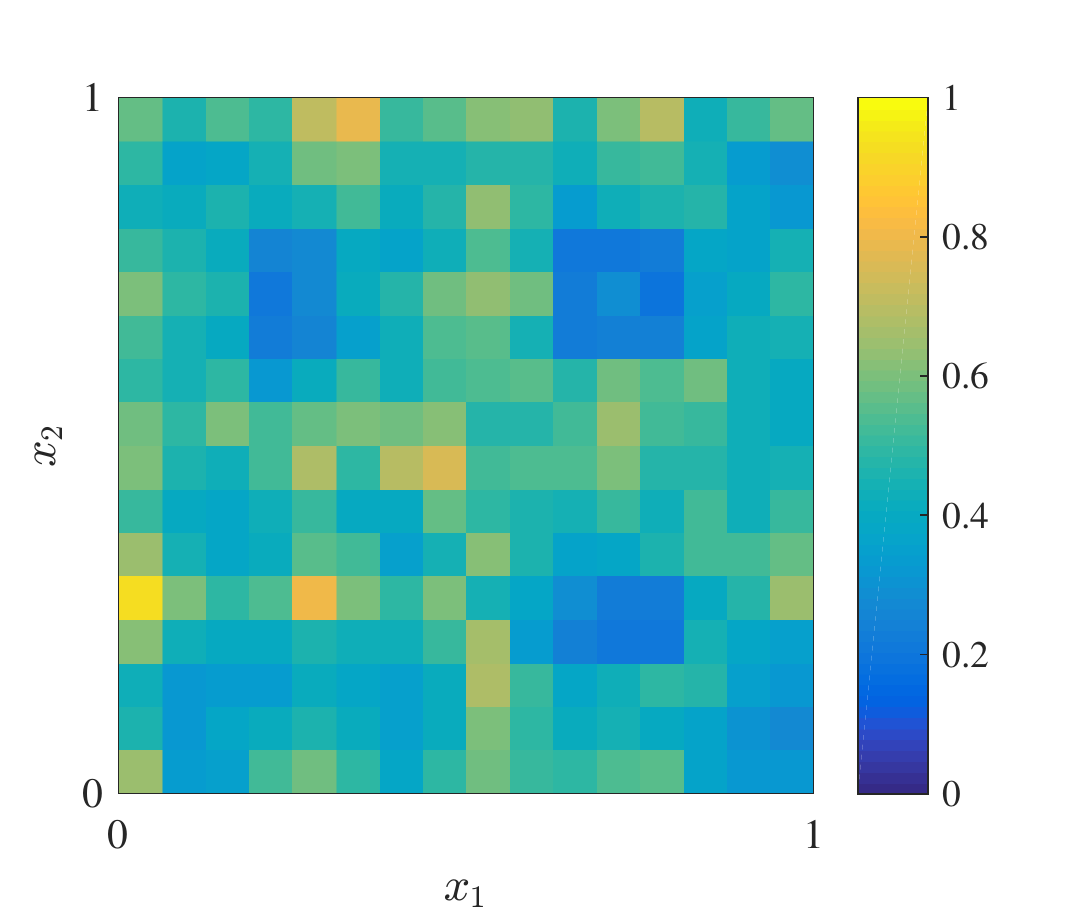}{0.49}}
\caption{Example C: Posterior moments for the 2-D elliptic PDE problem.}\label{fig:2d_mean}
\end{figure}

\begin{figure}[htbp]
\centering
\subfloat[Posterior samples $\theta^{(i)}$]{\pic{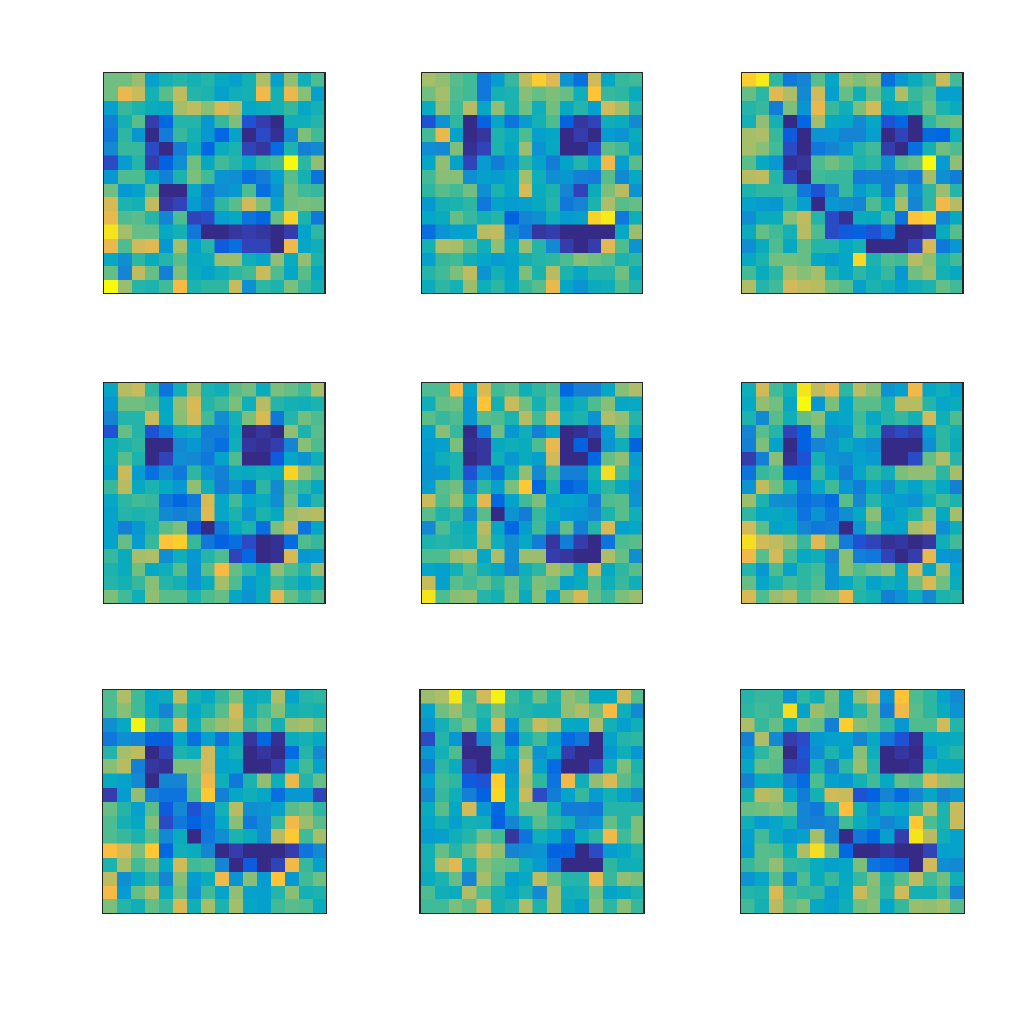}{0.6}}\\
\subfloat[Corresponding posterior predictive samples $f(\theta^{(i)})$]{\pic{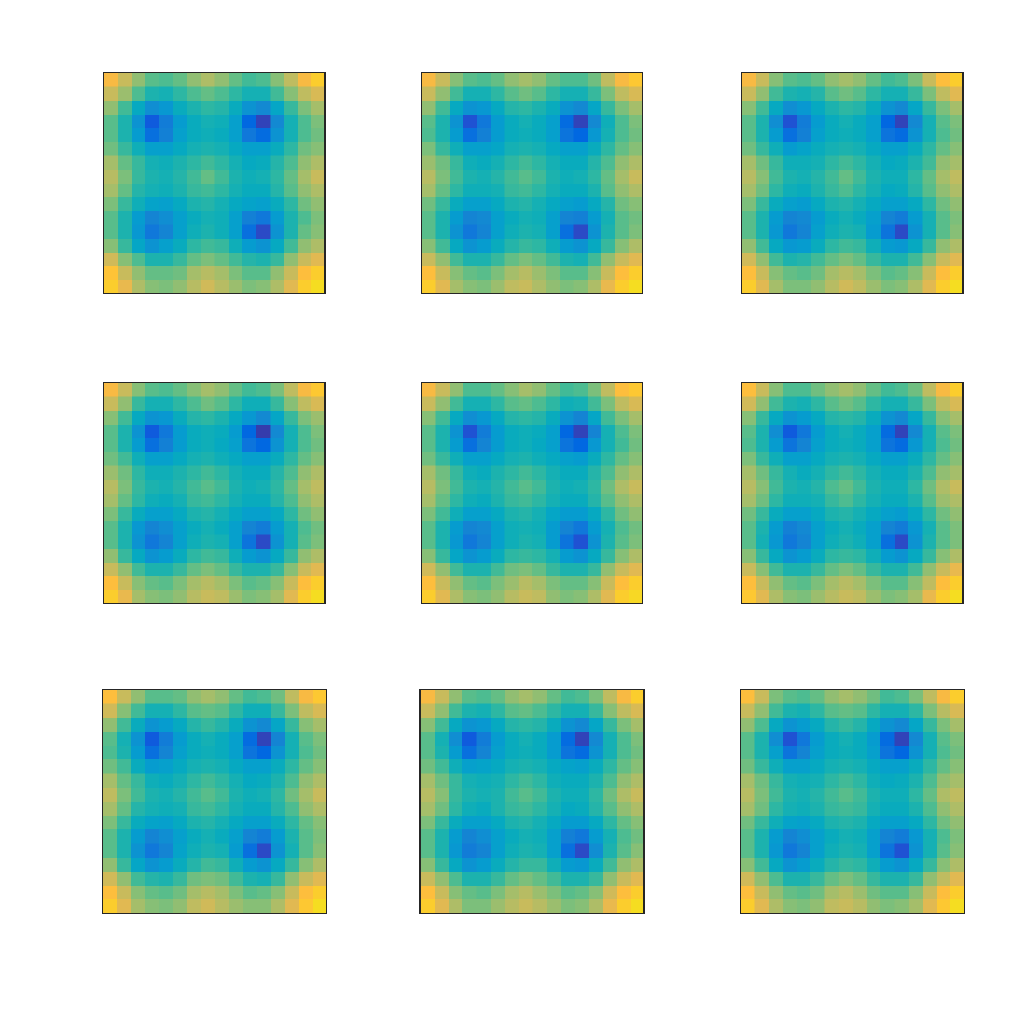}{0.6}}
\caption{Example C: Posterior samples and corresponding posterior predictive samples. The former have small-scale roughness and sample-to-sample variability, while the latter closely match the potential field measurements.}\label{fig:2d_samp}
\end{figure}

\resetcounters

\section{Concluding remarks}

We have extended RTO, an optimization-based sampling algorithm, to posterior distributions arising in Bayesian inverse problems with non-Gaussian priors. As a concrete example, we consider $l_1$-type priors such as TV and Besov \B priors. To transform the posterior into a form usable by RTO, we derive a deterministic map that transforms the prior to a standard Gaussian. We embed the RTO proposal into a Metropolis-Hastings algorithm to generate asymptotically exact samples from the transformed posterior, and then apply the transformation to obtain samples from the original posterior. Some assumptions are required for the probability density of the RTO proposal samples to be known and computable. We prove that these assumptions are satisfied for linear forward models and our transformation of $l_1$ priors. Numerical studies suggest that our method can be more efficient than standard MCMC algorithms, and that its sampling performance does not deteriorate as the parameter discretization is refined. We also successfully employ the algorithm for posterior sampling in a nonlinear inverse problem with a Besov \B prior in two spatial dimensions, suggesting that it is a promising and versatile computational approach for challenging problems.

\section*{Acknowledgments}
Z.\ Wang, A.\ Solonen, and Y.\ Marzouk acknowledge support from the eni-MIT Alliance research program. J.\ Bardsley was funded by the National Security Technologies, LLC, Site Directed Research and Development program.

\resetcounters

\appendix

\section{Proof of Lemma~\ref{lem:post}}\label{sec:proofpost}
We now derive the posterior density on $u$. In this appendix, we use more precise notation for clarity. Let $\pi_\Theta(w)$ be the prior density on $\theta$ evaluated at $\theta = w$, $\pi_U(w)$ be the prior density on $u$ evaluated at $u = w$, and so forth for the posterior densities. First, note that
\[ \pi_\Theta(\gd(u)) = \pi_U(u) \left|\frac{\partial}{\partial \theta} \gd^{-1}(\theta) \right|,\]
and thus
\begin{align*}
 \pi_{U|Y}(u|y) &= \pi_{\Theta|Y}(\gd(u)|y) \overbrace{\left|\frac{\partial}{\partial u} \gd(u)\right|}^{\left|\J {\gd}\right|} \\
 &\propto \exp{\left[-\half\left( \frac{f \circ \gd(u)-y}{\sigma_{\textrm{obs}}} \right)^2 \right]} \pi_\Theta(\gd(u)) \left|\frac{\partial}{\partial u} \gd(u)\right| \\
 &\propto \exp{\left[-\half\left( \frac{f \circ \gd(u)-y}{\sigma_{\textrm{obs}}} \right)^2 \right]} \pi_U(u) \underbrace{\left|\frac{\partial}{\partial \theta} \gd^{-1}(\theta) \right|}_{\left|\J {\gd}^{-1}\right|} \underbrace{\left|\frac{\partial}{\partial u} \gd(u) \right|}_{\left|\J {\gd}\right|} \\
&\propto \exp{\left[-\half\left( \frac{f \circ \gd(u)-y}{\sigma_{\textrm{obs}}} \right)^2 \right]} \pi_U(u)\\
&\propto \exp{\left[-\half\left( \frac{f \circ \gd(u)-y}{\sigma_{\textrm{obs}}} \right)^2 \right]} \exp{\left(-\half u^2 \right).}\\
\end{align*}
We note that in the third line, by the inverse function theorem, $\left|\frac{\partial}{\partial \theta} \gd^{-1}(\theta) \right|$ is the inverse of $\left|\frac{\partial}{\partial u} \gd(u)\right|$ and the two terms cancel.

\section{RTO proposal density and proof of Theorem~\ref{thm:linear}} \label{sec:assume}
First, we recall the assumptions under which the RTO proposal density in (\ref{eq:RTOproposal}) holds. Knowing the form of the proposal density is important because it allows us to ``correct'' the proposed samples and thus achieve exact sampling, for instance through the use of a Metropolis-Hastings scheme, or via importance sampling. The theorem that describes the required assumptions is found in \cite{rto} and restated below. 

\begin{assume}[Conditions for validity of the RTO proposal density]\label{assump}\\
\assumeitem \label{assume:1}
  $p(\theta|y) \propto \exp \left(-\half \Vert F(\theta)\Vert^2\right)$, where $\theta \in \R^n$.\\
\assumeitem \label{assume:2}
  $F : \R^n \rightarrow \R^{n+m}$ is a continuously differentiable function with Jacobian $\J{F}$.\\
\assumeitem \label{assume:3}
  $\J{F}(\theta) \in \R^{(n+m)\times n}$ has rank $n$ for every $\theta$ in the domain of $F$.\\
\assumeitem \label{assume:4}
  The matrix $\overline{Q}^\T \J{F}(\theta)$ is invertible for all $\theta$ in the domain of $F$, where 
  \[\J{F}(\overline{\theta}) = [\overline{Q},\widetilde{Q}]\left[\begin{array}{c}\overline{R}\\0\end{array}\right]\] 
  \makebox[2em][l]{~} is the QR factorization of $\J{F}(\overline{\theta})$, with $\overline{\theta}$ fixed.
\end{assume}

\begin{thm}[Proposal density for RTO \cite{rto}]\label{thm:rtoprop}
If Assumption~\ref{assump} holds, then the RTO algorithm described by Steps 1--7 of Algorithm~\ref{alg:rto} generates proposal samples distributed according to the probability density (\ref{eq:RTOproposal}).
\end{thm}

We now prove Theorem~\ref{thm:linear} by checking Assumptions~\ref{assume:1} to \ref{assume:4} for the transformed forward model $\tilde{f}(u)$. 

\begin{proof}[Proof of Theorem~\ref{thm:linear}]
If $\tilde{f}$ is continuously differentiable, then $\widetilde{F}$ is continuously differentiable. Thus Assumptions~\ref{assume:1} and \ref{assume:2} are automatically satisfied. Assumption~\ref{assume:3} is also satisfied since
\begin{equation*} 
J_{\widetilde{F}}(\theta) = \begin{bmatrix} I \\ J_{\tilde{f}}(\theta) \end{bmatrix},
\end{equation*}
and regardless of $J_{\tilde{f}}(\theta)$, the columns of $\J{\widetilde{F}}(\theta)$ are linearly independent due to the identity matrix in the first $n$ rows of $\J{\widetilde{F}}(\theta)$. 

To show that Assumption~\ref{assume:4} holds, we use the form of the transformed forward model. Let the original linear forward model be $f(\theta) = A \theta$. Then the transformed forward model is
\[\tilde{f}(u) = f(D^{-1} g(u)) = AD^{-1} g(u).\]
Following the computations used to obtain (\ref{eq:TransPost}), the posterior on $u$ takes the form 
\begin{align*}
p(u|y) 
&\propto \exp{\left[ -\half (\tilde{f}(u) - y)^\T \Gamma_\textrm{obs}^{-1} (\tilde{f}(u) - y)  \right]} \exp{\left( -\half u^\T u  \right)}\\
&= \exp{\left(-\half \left\|\widetilde{F}(u)\right\|^2\right)},
\end{align*}
where
\begin{align*}
{\widetilde{F}}(u)&=\begin{bmatrix}u\\ \Gamma_\textrm{obs}^{-1/2} (\tilde{f}(u)-y) \end{bmatrix}, &
J_{\widetilde{F}}(u) & = \begin{bmatrix}I\\ \Gamma_\textrm{obs}^{-1/2} AD^{-1} J_g(u) \end{bmatrix}.
\end{align*}

Assumption~\ref{assume:4} requires that the matrix $\overline{Q}^\T \J{\widetilde{F}}(u)$ be invertible for all $u$ in the domain of $\widetilde{F}$. For any $u_1 \in \R^n$ and $u_2 \in \R^n$,
\begin{align*}
\J{\widetilde F}(u_1)^\T \J{\widetilde F}(u_2) &=  I +  J_g(u_1) D^{-\T} A^\T \Gamma_\textrm{obs}^{-1} A D^{-1} J_g(u_2) \\
&= J_g(u_1) \left( J_g(u_1)^{-1} J_g(u_2)^{-1} +  D^{-\T} A^\T \Gamma_\textrm{obs}^{-1} A D^{-1} \right) J_g(u_2).
\end{align*}
$J_g(u)$ is a positive diagonal matrix for any $u$, and $D^{-\T} A^\T \Gamma_\textrm{obs}^{-1} A D^{-1}$ is symmetric positive semi-definite. Then, the middle matrix is symmetric positive definite. Thus, $\J{\widetilde F}(u_1)^\T \J{\widetilde F}(u_2)$ is the product of three invertible matrices and is therefore invertible.

$\overline{Q}$ is obtained from the thin QR-decomposition of $\J{\widetilde F}(\overline{u})$, where $\overline{u}$ is the mode of the posterior defined on $u$. It follows that $\J{\widetilde F}(u)^\T \overline{Q} = \J{\widetilde F}(u)^\T \J{\widetilde F}(\overline{u}) \overline{R}^{-1}$ is invertible for any $u \in \R^n$. This shows that Assumption~\ref{assume:4} holds. Hence, Assumptions~\ref{assume:1} to \ref{assume:4} hold for the transformed forward model $\tilde{f}(u)$ and Theorem~\ref{thm:rtoprop} yields Theorem~\ref{thm:linear}.
\end{proof}

\section{Besov space priors as \texorpdfstring{$l_1$}{l1}-type priors} \label{sec:besov}
Following \cite{sparsity}, we start with a wavelet function $\psi \in \mathcal{L}_2([0,1])$ defined such that the family of functions 
\[ \psi_{j,k} (x) = 2^{\frac{j}{2}} \psi \big(2^j x -k\big), \spaces j, k \in \Z_+, \spaces 0 \le k \le 2^j-1,\]
is an orthonormal basis for $\mathcal{L}_2([0,1])$. One example of such a function is the Haar wavelet,
\[ \psi_\textrm{Haar} = \left\{\begin{matrix}1  & \text{ when } 0 < x < \half\\ -1 & \text{ when } \half < x < 1 \end{matrix}\right..\]
With a wavelet and corresponding basis, we can represent functions by the expansion
\[ f(x) = c_0 + \sum_{j=0}^\infty \sum_{k=0}^{2^j-1} w_{j,k} \psi_{j,k}(x), \spaces c_0 \defeq \int_0^1 f(x) dx, \spaces w_{j,k} \defeq \int_0^1 f(x) \psi_{j,k}(x) dx.\]
The Besov space $B_{p,q}^s ([0,1])$ contains functions over the interval $[0,1]$ with a finite Besov $B_{p,q}^s([0,1])$ norm, defined as  
\[\|f\|_{B_{p,q}^s ([0,1])} \defeq \left( |c_0|^q + \sum_{j=0}^\infty 2^{jq\big(s+\half-\frac{1}{p}\big)}\left( \sum_{n=0}^{2^j-1} |w_{j,k}|^p \right)^\frac{q}{p}\right)^\frac{1}{q}, \]
where $s \in \R$ and $p,q \ge 1$ are properties of the space, and $c_0, w_{j,k} \in \R$ are the coefficients of the expansion. The discrete Besov $B_{p,q}^s$ space norm, defined for a vector $\theta \in \R^n$ of size $n = 2^l$, is
\[ \|\theta\|_{B_{p,q}^s} \defeq  \left( |\hat c_0|^q + \sum_{j=0}^l 2^{jq\big(s+\half-\frac{1}{p}\big)}\left( \sum_{n=0}^{2^j-1} |\hat w_{j,k}|^p \right)^\frac{q}{p}\right)^\frac{1}{q}, \]
where $\hat c_0, \hat w_{j,k} \in \R$ are the coefficients 
\begin{align*}
\hat c_0 &= \frac{1}{n} \theta^\T \hat\phi_{0,0}, & \hat w_{j,k}&= \frac{1}{n} \theta^\T \hat\psi_{j,k}, \\
\hat\phi_{0,0} &= [1,\cdots, 1]^\T, & \hat\psi_{j,k}&= \left[\psi_{j,k}\left(\frac{1}{2n}\right), \psi_{j,k}\left(\frac{3}{2n}\right), \cdots, \psi_{j,k}\left(\frac{2n-1}{2n}\right) \right]^\T.
\end{align*}
Note that when $\theta \in \R^n$ is a discretization of the continuous function $f : [0,1] \rightarrow \R$,
\[ \theta = \left[f\left(\frac{1}{2n}\right), f\left(\frac{3}{2n}\right), \cdots, f\left(\frac{2n-1}{2n}\right) \right]^\T. \]
Then, the discrete norm $\|\theta\|_{B_{p,q}^s}$ is an approximation to the continuous norm $\|f\|_{B_{p,q}^s}$. When $p = q = 1$, the discrete Besov $B_{p,q}^s$ space norm becomes
\begin{align*} \|\theta\|_{B_{1,1}^s} &= |\hat c_0| + \sum_{j=0}^l \sum_{h=0}^{2^j-1} 2^{j\left(s-\half \right)} |\hat w_{j,k}| \\
&= \| W B\theta \|_1, \end{align*}
where the matrix $W \in \R^{n \times n}$ is diagonal with
\[W_{1,1} = \frac{1}{\sqrt{n}} \spaces \text{ and } \spaces W_{i,i} = \frac{1}{\sqrt{n}}2^{j\left(s-\half \right)} \text{ when }  2^j + 1 \le i \le 2^{j+1},  \spaces\]
and the matrix $B \in \R^{n \times n}$ is unitary with
\[B = \frac{1}{\sqrt n} \bigg[ \hat \phi_{0,0} \;\;\; \hat \psi_{0,0} \;\;\; \hat \psi_{1,0} \;\;\; \hat \psi_{1,1} \;\;\; \hat \psi_{2,0} \;\;\; \cdots \bigg]^\T.\]
Thus, we can write the Besov \B space prior in the form of (\ref{eq:l1prior}) by
\[ p(\theta) \defeq \exp\left(-\lambda \|\theta\|_{B_{1,1}^s}\right) = \exp\Big(-\lambda \| D \theta \|_1 \Big), \]
where $D=WB$, with $W$ and $B$ defined as above.

\section{Pointwise variance of Besov priors with Haar wavelets} \label{sec:priorvar}
Let $f$ be a random function distributed according to the Besov \B prior, using Haar wavelets. $f$ can be represented by the expansion
\[ f(x) = c_0 + \sum_{j=0}^{\infty} \sum_{k=0}^{2^j-1} w_{j,h} \phi_{j,h}(x) . \]
Fix any point $x^*$ and consider the random variable $f(x^*)$. For each level $j$, there is only one basis, $\phi_{j,h^*}$, that has a support containing $x^*$, where $h^*$ depends on both $x^*$ and $j$. Also, the magnitude of $\phi_{j,h^*}$ evaluated at $x^*$ is $2^{\frac{j}{2}}$. Thus,
\[ f(x^*) = c_0 + \sum_{j=0}^\infty \pm 2^{\frac{j}{2}} w_{j,h^*} \]
where, due to the Besov \B space prior,
\[ c_0 \sim \text{Laplace}(0,1), \spaces w_{j,h^*} \sim \text{Laplace}\big(0, 2^{-j(s-\half)}\big).\]
We sum the variance contribution from each coefficient.
\[ \Var \big[ f(x^*) \big] = 2 + \sum_{j=0}^\infty 2^j \big(2 \cdot 2^{-2j(s -\half)} \big) = 2 \bigg(1 + \sum_{j=0}^\infty  2^{-2j(s-1)} \bigg) \] 
Hence, the pointwise variance is finite when $s>1$ and does not converge when $s = 1$.

\bibliographystyle{siamplain}

\bibliography{references}

\begin{thebibliography}{10}

\bibitem{amala}
{\sc Y.~F. Atchad{\'e}}, {\em An adaptive version for the {M}etropolis adjusted
  {L}angevin algorithm with a truncated drift}, Methodology and Computing in
  applied Probability, 8 (2006), pp.~235--254.

\bibitem{BarLaplace}
{\sc J.~M. Bardsley}, {\em {L}aplace-distributed increments, the {L}aplace
  prior, and edge-preserving regularization}, Journal of Inverse and Ill-Posed
  Problems, 20 (2012), pp.~271--–285.

\bibitem{rto}
{\sc J.~M. Bardsley, A.~Solonen, H.~Haario, and M.~Laine}, {\em
  Randomize-then-optimize: A method for sampling from posterior distributions
  in nonlinear inverse problems}, SIAM Journal on Scientific Computing, 36
  (2014), pp.~A1895--A1910.

\bibitem{CalSom}
{\sc D.~Calvetti and E.~Somersalo}, {\em An Introduction to {B}ayesian
  Scientific Computing: Ten Lectures on Subjective Computing}, vol.~2, Springer
  Science \& Business Media, 2007.

\bibitem{generalimplicit}
{\sc A.~Chorin, M.~Morzfeld, and X.~Tu}, {\em Implicit particle filters for
  data assimilation}, Communications in Applied Mathematics and Computational
  Science, 5 (2010), pp.~221--240.

\bibitem{pCN}
{\sc S.~L. Cotter, G.~O. Roberts, A.~M. Stuart, and D.~White}, {\em {MCMC}
  methods for functions: modifying old algorithms to make them faster},
  Statistical Science, 28 (2013), pp.~424--446.

\bibitem{dili}
{\sc T.~Cui, K.~J. Law, and Y.~M. Marzouk}, {\em Dimension-independent
  likelihood-informed {MCMC}}, Journal of Computational Physics, 304 (2016),
  pp.~109--137.

\bibitem{besov}
{\sc M.~Dashti, S.~Harris, and A.~Stuart}, {\em {B}esov priors for {B}ayesian
  inverse problems}, Inverse Problems and Imaging, 6 (2012), pp.~183--200.

\bibitem{daubechies}
{\sc I.~Daubechies}, {\em Ten Lectures on Wavelets}, vol.~61, SIAM, 1992.

\bibitem{DauDefDeM}
{\sc I.~Daubechies, M.~Defrise, and C.~De~Mol}, {\em An iterative thresholding
  algorithm for linear inverse problems with a sparsity constraint},
  Communications on Pure and Applied Mathematics, 57 (2004), pp.~1413--1457.

\bibitem{GamLop}
{\sc D.~Gamerman and H.~F. Lopes}, {\em {M}arkov Chain {M}onte {C}arlo:
  Stochastic Simulation for {B}ayesian inference}, Chapman and Hall/CRC Press,
  2006.

\bibitem{gibbs}
{\sc A.~E. Gelfand and A.~F. Smith}, {\em Sampling-based approaches to
  calculating marginal densities}, Journal of the American Statistical
  Association, 85 (1990), pp.~398--409.

\bibitem{bayesiandata}
{\sc A.~Gelman, J.~B. Carlin, H.~S. Stern, and D.~B. Rubin}, {\em {B}ayesian
  Data Analysis}, vol.~2, Chapman \& Hall/CRC Boca Raton, FL, USA, 2014.

\bibitem{badparallelGeyer}
{\sc C.~J. Geyer}, {\em Practical {M}arkov chain {M}onte {C}arlo}, Statistical
  Science,  (1992), pp.~473--483.

\bibitem{mcmcinpractice}
{\sc W.~R. Gilks, S.~Richardson, and D.~J. Spiegelhalter}, {\em {M}arkov Chain
  {M}onte {C}arlo in Practice}, Chapman \& Hall/CRC Interdisciplinary
  Statistics, 1996.

\bibitem{riemann}
{\sc M.~Girolami and B.~Calderhead}, {\em Riemann manifold {L}angevin and
  {H}amiltonian {M}onte {C}arlo methods}, Journal of the Royal Statistical
  Society: Series B (Statistical Methodology), 73 (2011), pp.~123--214.

\bibitem{green}
{\sc P.~J. Green}, {\em Reversible jump {M}arkov chain {M}onte {C}arlo
  computation and {B}ayesian model determination}, Biometrika, 82 (1995),
  pp.~711--732.

\bibitem{dram}
{\sc H.~Haario, M.~Laine, A.~Mira, and E.~Saksman}, {\em {DRAM}: efficient
  adaptive {MCMC}}, Statistics and Computing, 16 (2006), pp.~339--354.

\bibitem{adaptm}
{\sc H.~Haario, E.~Saksman, and J.~Tamminen}, {\em An adaptive {M}etropolis
  algorithm}, Bernoulli,  (2001), pp.~223--242.

\bibitem{hastings}
{\sc W.~K. Hastings}, {\em {M}onte {C}arlo sampling methods using {M}arkov
  chains and their applications}, Biometrika, 57 (1970), pp.~97--109.

\bibitem{nuts}
{\sc M.~D. Hoffman and A.~Gelman}, {\em The {N}o-{U}-turn sampler: adaptively
  setting path lengths in {H}amiltonian {M}onte {C}arlo.}, Journal of Machine
  Learning Research, 15 (2014), pp.~1593--1623.

\bibitem{vartrans}
{\sc L.~T. Johnson and C.~J. Geyer}, {\em Variable transformation to obtain
  geometric ergodicity in the random-walk {M}etropolis algorithm}, The Annals
  of Statistics, 40 (2012), pp.~3050--3076.

\bibitem{KaipioSomersalo}
{\sc J.~Kaipio and E.~Somersalo}, {\em Statistical and Computational Inverse
  Problems}, vol.~160, Springer Science \& Business Media, 2006.

\bibitem{KaiKolSomVau}
{\sc J.~P. Kaipio, V.~Kolehmainen, E.~Somersalo, and M.~Vauhkonen}, {\em
  Statistical inversion and {M}onte {C}arlo sampling methods in electrical
  impedance tomography}, Inverse Problems, 16 (2000), p.~1487.

\bibitem{sparsity}
{\sc V.~Kolehmainen, M.~Lassas, K.~Niinim{\"a}ki, and S.~Siltanen}, {\em
  Sparsity-promoting {B}ayesian inversion}, Inverse Problems, 28 (2012),
  p.~025005.

\bibitem{discretization-invariant}
{\sc M.~Lassas, E.~Saksman, and S.~Siltanen}, {\em Discretization-invariant
  {B}ayesian inversion and {B}esov space priors}, Inverse Problems and Imaging,
  3 (2009), pp.~87--122.

\bibitem{tvnotok}
{\sc M.~Lassas and S.~Siltanen}, {\em Can one use total variation prior for
  edge-preserving {B}ayesian inversion?}, Inverse Problems, 20 (2004), p.~1537.

\bibitem{mcscientific}
{\sc J.~S. Liu}, {\em {M}onte {C}arlo Strategies in Scientific Computing},
  Springer Science \& Business Media, 2008.

\bibitem{felix}
{\sc F.~Lucka}, {\em Fast {M}arkov chain {M}onte {C}arlo sampling for sparse
  {B}ayesian inference in high-dimensional inverse problems using {L}1-type
  priors}, Inverse Problems, 28 (2012), p.~125012.

\bibitem{MarWilBurGha}
{\sc J.~Martin, L.~C. Wilcox, C.~Burstedde, and O.~Ghattas}, {\em A stochastic
  {N}ewton {MCMC} method for large-scale statistical inverse problems with
  application to seismic inversion}, SIAM Journal on Scientific Computing, 34
  (2012), pp.~A1460--A1487.

\bibitem{mapsHandbookChapter}
{\sc Y.~Marzouk, T.~Moselhy, M.~Parno, and A.~Spantini}, {\em Sampling via
  measure transport: An introduction}, in Handbook of Uncertainty
  Quantification, R.~Ghanem, D.~Higdon, and H.~Owhadi, eds., Springer, 2016.

\bibitem{convdifflim}
{\sc J.~C. Mattingly, N.~S. Pillai, and A.~M. Stuart}, {\em Diffusion limits of
  the random walk {M}etropolis algorithm in high dimensions}, The Annals of
  Applied Probability, 22 (2012), pp.~881--930.

\bibitem{convrates}
{\sc K.~L. Mengersen and R.~L. Tweedie}, {\em Rates of convergence of the
  {H}astings and {M}etropolis algorithms}, The Annals of Statistics, 24 (1996),
  pp.~101--121.

\bibitem{metropolis}
{\sc N.~Metropolis, A.~W. Rosenbluth, M.~N. Rosenbluth, A.~H. Teller, and
  E.~Teller}, {\em Equation of state calculations by fast computing machines},
  The Journal of Chemical Physics, 21 (1953), pp.~1087--1092.

\bibitem{implicit}
{\sc M.~Morzfeld, X.~Tu, E.~Atkins, and A.~J. Chorin}, {\em A random map
  implementation of implicit filters}, Journal of Computational Physics, 231
  (2012), pp.~2049--2066.

\bibitem{tarekTM}
{\sc T.~A. Moselhy and Y.~M. Marzouk}, {\em {B}ayesian inference with optimal
  maps}, Journal of Computational Physics, 231 (2012), pp.~7815–--7850.

\bibitem{MueSil}
{\sc J.~L. Mueller and S.~Siltanen}, {\em Linear and Nonlinear Inverse Problems
  with Practical Applications}, vol.~10, SIAM, 2012.

\bibitem{hmc}
{\sc R.~M. Neal}, {\em {MCMC} using {H}amiltonian dynamics}, Handbook of
  {M}arkov Chain {M}onte {C}arlo, 2 (2011), pp.~113--162.

\bibitem{NicFox}
{\sc G.~K. Nicholls and C.~Fox}, {\em Prior modeling and posterior sampling in
  impedance imaging}, in SPIE's International Symposium on Optical Science,
  Engineering, and Instrumentation, International Society for Optics and
  Photonics, 1998, pp.~116--127.

\bibitem{Oliver_arXiv:1507.08563}
{\sc D.~S. Oliver}, {\em Metropolized {R}andomized {M}aximum {L}ikelihood for
  sampling from multimodal distributions}, arXiv preprint arXiv:1507.08563,
  (2015).

\bibitem{rml}
{\sc D.~S. Oliver, N.~He, and A.~C. Reynolds}, {\em Conditioning permeability
  fields to pressure data}, in ECMOR V-5th European Conference on the
  Mathematics of Oil Recovery, EAGE, September 1996.

\bibitem{mattMCMC}
{\sc M.~Parno and Y.~Marzouk}, {\em Transport map accelerated {M}arkov chain
  {M}onte {C}arlo}, arXiv preprint arXiv:1412.5492,  (2014).

\bibitem{mcstatmodel}
{\sc C.~Robert and G.~Casella}, {\em {M}onte {C}arlo Statistical Methods},
  Springer Science \& Business Media, 2013.

\bibitem{convrwm}
{\sc G.~O. Roberts, A.~Gelman, and W.~R. Gilks}, {\em Weak convergence and
  optimal scaling of random walk {M}etropolis algorithms}, The Annals of
  Applied Probability, 7 (1997), pp.~110--120.

\bibitem{mala}
{\sc G.~O. Roberts and R.~L. Tweedie}, {\em Exponential convergence of
  {L}angevin distributions and their discrete approximations}, Bernoulli,
  (1996), pp.~341--363.

\bibitem{tv}
{\sc L.~I. Rudin, S.~Osher, and E.~Fatemi}, {\em Nonlinear total variation
  based noise removal algorithms}, Physica D: Nonlinear Phenomena, 60 (1992),
  pp.~259--268.

\bibitem{StuartActa}
{\sc A.~M. Stuart}, {\em Inverse problems: a {B}ayesian perspective}, Acta
  Numerica, 19 (2010), pp.~451--559.

\bibitem{Tarantola}
{\sc A.~Tarantola}, {\em Inverse Problem Theory and Methods for Model Parameter
  Estimation}, Other Titles in Applied Mathematics, SIAM, 2005.

\bibitem{Vogel}
{\sc C.~R. Vogel}, {\em Computational Methods for Inverse Problems}, vol.~23,
  SIAM, 2002.

\bibitem{wolff}
{\sc U.~Wolff and {ALPHA Collaboration}}, {\em {M}onte {C}arlo errors with less
  errors}, Computer Physics Communications, 156 (2004), pp.~143--153.

\end{thebibliography}

\end{document}